\documentclass[journal]{IEEEtran}
\usepackage{cite}

\ifCLASSINFOpdf
  \usepackage[pdftex]{graphicx}
\else
\fi
\usepackage[cmex10]{amsmath}
\usepackage{algorithmic}
\usepackage[ruled, vlined]{algorithm2e}

\usepackage{array}

\usepackage{cancel}
\usepackage{stfloats}
\usepackage{url}

\usepackage[tight,footnotesize]{subfigure}
\usepackage{labtex}
\usepackage{enumerate}

\newcommand{\s}{\mathcal{S}}
\newtheorem{remark}{Remark}
\newcommand{\RHinf}{\mathcal{RH}_\infty}

\newcommand{\ttf}[1]{\boldsymbol{#1}}
\newcommand{\tf}[1]{\mathbf{#1}}

\newcommand{\C}{\mathcal{C}}

\newcommand{\FT}{\mathcal F_T}
\renewcommand{\supp}[1]{\mathrm{supp}\left(#1\right)}

\newenvironment{proof}
{\begin{IEEEproof}}{%
    \end{IEEEproof}\ignorespacesafterend
}
\newcommand{\minimize}[1]{\underset{#1}{\text{minimize}}}

\usepackage[normalem]{ulem} 
\usepackage{color}
\newcommand{\rev}[1]{\textcolor{black}{#1}}
\newcommand{\revsecond}[1]{\textcolor{black}{#1}}

\begin{document}
%
\title{A System Level Approach to Controller Synthesis}
%
%
%

\author{Yuh-Shyang~Wang,~\IEEEmembership{Member,~IEEE,}
        Nikolai~Matni,~\IEEEmembership{Member,~IEEE,}
        and~John~C.~Doyle
\thanks{This paper was presented in part at IEEE American Control Conference, June 4 - 6, 2014; in part at 52nd Annual Allerton Conference on Communication, Control, and Computing, September 30 - October 3, 2014; in part at 53rd IEEE Conference on Decision and Control, Los Angeles, CA, USA, December 15 - 17, 2014; in part at 5th IFAC Workshop on Distributed Estimation and Control in Networked Systems, September 10 - 11, 2015; in part at IEEE American Control Conference, July 6 - 8, 2016; and in part at IEEE American Control Conference, May 24 - 26, 2017 \cite{2014_Wang_ACC,2014_Wang_CDC,2014_Wang_Allerton,2015_Wang_H2,2015_Wang_Reg,2015_Wang_LDKF,WMD_ACC17}. This work was supported by Air Force Office of Scientific Research and National Science Foundation and gifts from Huawei and Google.}
\thanks{Y.-S. Wang is with the Control and Optimization Group, GE Global Research Center, Niskayuna, NY 12309 USA (e-mail: yuh-shyang.wang@ge.com).}
\thanks{N. Matni is with the Department of Electrical Engineering and Computer Sciences, UC Berkeley, Berkeley, CA 94720 USA (e-mail: nmatni@berkeley.edu).}
\thanks{J. C. Doyle is with the Department of Control and Dynamical Systems, California Institute of Technology, Pasadena, CA 91125 USA (e-mail: doyle@caltech.edu).}
}

%
%

\markboth{To Appear in IEEE Transactions on Automatic Control,~Vol.~XX, No.~XX, October~2019}%
{Shell \MakeLowercase{\textit{et al.}}: A System Level Approach to Controller Synthesis}
%



\maketitle

\begin{abstract}
Biological and advanced cyberphysical control systems often have limited, sparse, uncertain, and distributed communication and computing in addition to sensing and actuation.  Fortunately, the corresponding plants and performance requirements are also sparse and structured, and this must be exploited to make constrained controller design feasible and tractable.  We introduce a new ``system level" (SL) approach involving three complementary SL elements. System Level Parameterizations (SLPs) \revsecond{provide an alternative to} the Youla parameterization of all stabilizing controllers and the responses they achieve, and combine with System Level Constraints (SLCs) to parameterize the largest known class of constrained stabilizing controllers that admit a convex characterization, generalizing quadratic invariance (QI).  SLPs also lead to a generalization of detectability and stabilizability, suggesting the existence of a rich separation structure, that when combined with SLCs, is naturally applicable to structurally constrained controllers and systems.  We further provide a catalog of useful SLCs, most importantly including sparsity, delay, and locality constraints on both communication and computing internal to the controller, and external system performance. \revsecond{Finally, we formulate System Level Synthesis (SLS) problems, which define the broadest known class of constrained optimal control problems that can be solved using convex programming.} 

\end{abstract}


%
\IEEEpeerreviewmaketitle


 
\subsubsection*{Preliminaries \& Notation}
We use lower and upper case Latin letters such as $x$ and $A$ to denote vectors and matrices, respectively, and lower and upper case boldface Latin letters such as $\tf x$ and $\tf G$ to denote signals and transfer matrices, respectively.  We use calligraphic letters such as $\s$ to denote sets. In the interest of clarity, we work with discrete time linear time invariant systems, but unless stated otherwise, all results extend naturally to the continuous time setting.  We use standard definitions of the Hardy spaces $\mathcal{H}_2$ and $\mathcal{H}_\infty$, and denote their restriction to the set of real-rational proper transfer matrices by $\mathcal{RH}_2$ and $\RHinf$.  We use $G[i]$ to denote the $i$th spectral component of a transfer function $\tf G$, i.e., $\tf G(z) = \sum_{i=0}^{\infty} \frac{1}{z^i} G[i]$ for $| z | > 1$.  Finally, we use $\FT$ to denote the space of finite impulse response (FIR) transfer matrices with horizon $T$, i.e., $\FT := \{ \tf G \in \RHinf \, | \, \tf G = \sum_{i=0}^T\frac{1}{z^i}G[i]\}$.

\section{Introduction}
%
%
%
%
\IEEEPARstart{T}{he} foundation of many optimal controller synthesis procedures is a parameterization of all internally stabilizing controllers, and the responses that they achieve, over which relevant performance measures can be easily optimized.  \revsecond{For finite dimensional linear-time-invariant (LTI) systems, the class of internally stabilizing LTI feedback controllers is characterized by the celebrated Youla parameterization \cite{1976_Youla_parameterization} and the closely related factorization approach \cite{Factorization}. In \cite{1976_Youla_parameterization}, the authors showed that the Youla parameterization defines an isomorphism between a stabilizing controller and the resulting closed loop system response from sensors to actuators -- therefore rather than synthesizing the controller itself, this system response (or Youla parameter) could be directly optimized. This allowed for the incorporation of customized design specifications on the closed loop system into the controller design process via convex optimization \cite{Boyd_closed} or interpolation \cite{Dahleh}. Subsequently, analogous parameterizations of stabilizing controllers for more general classes of systems were developed: notable examples include the polynomial approach \cite{Polynomial} for generalized Rosenbrock systems \cite{Rosenbrock} and the behavioral approach \cite{Willems,Behavior_I, Behavior_II, Behavior_para} for linear differential systems. These results illustrate the power and generality of Youla parameterization and factorization approaches to optimal control in the centralized setting.} Together with state-space methods, they played a major role in shifting controller synthesis from an ad hoc, loop-at-a-time tuning process to a principled one with well defined notions of optimality, and in the LTI setting, paved the way for the foundational results of robust and optimal control that would follow \cite{1989_DGKF}.

\rev{However, as control engineers shifted their attention from centralized to distributed optimal control, it was observed that the parameterization approaches that were so fruitful in the centralized setting were no longer directly applicable.} In contrast to centralized systems, modern cyber-physical systems (CPS) are large-scale, physically distributed, and interconnected.  Rather than a logically centralized controller, these systems are composed of several sub-controllers, each equipped with their own sensors and actuators -- these sub-controllers then exchange locally available information (such as sensor measurements or applied control actions) via a communication network.  \revsecond{These information sharing constraints make the corresponding distributed optimal controller synthesis problem challenging to solve \cite{1972_Ho_info,2012_Mahajan_Info_survey,2006_Rotkowitz_QI_TAC,2002_Bamieh_spatially_invariant,2005_Bamieh_spatially_invariant,2013_Nayyar_common_info}. In particular, imposing such structural constraints on the controller can lead to optimal control problems that are NP-hard \cite{1968_Witsenhausen_counterexample,1984_Tsitsiklis_NP_hard}.}

Despite these technical and conceptual challenges, a body of work \cite{2002_Bamieh_spatially_invariant, 2004_QI, 2004_Dullerud, 2005_Bamieh_spatially_invariant, 2006_Rotkowitz_QI_TAC, 2012_Mahajan_Info_survey, 2013_Nayyar_common_info} that began in the early 2000s, and that culminated with the introduction of quadratic invariance (QI) in the seminal paper \cite{2006_Rotkowitz_QI_TAC}, showed that for a large class of practically relevant LTI systems, such internal structure could be integrated with the Youla parameterization and still preserve the convexity of the optimal controller synthesis task.  Informally, a system is quadratically invariant if sub-controllers are able to exchange information with each other faster than their control actions propagate through the CPS \cite{2010_Rotkowitz_QI_delay}.  Even more remarkable is that this condition is tight, in the sense that QI is a necessary \cite{2014_Lessard_convexity} and sufficient \cite{2006_Rotkowitz_QI_TAC} condition for subspace constraints (defined by, for example, communication delays) on the controller to be enforceable via convex constraints on the Youla parameter. \revsecond{The identification of QI triggered an explosion of results in distributed optimal controller synthesis \cite{2012_Lessard_two_player,2010_Shah_H2_poset,2013_Lamperski_H2,2013_Lessard_structure,2013_Scherer_Hinf,2014_Lessard_Hinf,2014_Matni_Hinf,2014_Tanaka_Triangular,2014_Lamperski_state} -- these results showed that the robust and optimal control methods that proved so powerful for centralized systems could be ported to the distributed setting.} \rev{As far as we are aware, no such results exist for the more general classes of systems considered in \cite{Polynomial,Willems,Behavior_I,Behavior_II,Behavior_para}.}

\rev{However, a fact that is not emphasized in the distributed optimal control literature is that distributed controllers are actually more complex to synthesize and implement than their centralized counterparts.\footnote{For example, see the solutions presented in \cite{2012_Lessard_two_player,2010_Shah_H2_poset,2013_Lamperski_H2,2013_Lessard_structure,2013_Scherer_Hinf,2014_Lessard_Hinf,2014_Matni_Hinf,2014_Tanaka_Triangular,2014_Lamperski_state} and the message passing implementation suggested in \cite{2014_Lamperski_state}.} } \rev{In particular, a major limitation of the QI framework is that, for strongly connected systems,\footnote{We say that a plant is strongly connected if the state of any subsystem can eventually alter the state of all other subsystems.} it cannot provide a convex characterization of localized controllers, in which local sub-controllers only access a subset of system-wide measurements (c.f., Section \ref{sec:limitation} and \ref{sec:sparsity}). This need for global exchange of information between sub-controllers is a limiting factor in the scalability of the synthesis and implementation of these distributed optimal controllers.}

\revsecond{Motivated by this issue, we propose a novel parameterization of internally stabilizing controllers and the closed loop responses that they achieve, providing an alternative to the QI framework for constrained optimal controller synthesis.} Specifically, rather than directly designing only the feedback loop between sensors and actuators, as in the Youla framework, we propose directly designing the entire closed loop response of the system, as captured by the maps from process and measurement disturbances to control actions and states.  As such, we call the proposed method a System Level Approach (SLA) to controller synthesis, which is composed of three elements: System Level Parameterizations (SLPs), System Level Constraints (SLCs) and System Level Synthesis (SLS) problems.  Further, in contrast to the QI framework, which seeks to impose structure on the input/output map between sensor measurements and control actions, the SLA imposes structural constraints on the system response itself, and shows that this structure carries over to the \emph{internal realization} of the corresponding controller.  It is this conceptual shift from structure on the input/output map to the internal realization of the controller that allows us to expand the class of structured controllers that admit a convex characterization, and in doing so, vastly increase the scalability of distributed optimal control methods.  We summarize our main contributions below. 

\subsection{Contributions}
This paper presents novel theoretical and computational contributions to the area of constrained optimal controller synthesis. In particular, we
\begin{itemize}
\item define and analyze the system level approach to controller synthesis, which is built around novel SLPs of all stabilizing controllers and the closed loop responses that they achieve;
\item show that SLPs allow us to constrain the closed loop response of the system to lie in arbitrary sets: we call such constraints on the system SLCs.  If these SLCs admit a convex representation, then the resulting set of constrained system responses admits a convex representation as well;
\item show that such constrained system responses can be used to directly implement a controller achieving them -- in particular, any SLC imposed on the system response imposes a corresponding SLC on the internal structure of the resulting controller;
\item show that the set of constrained stabilizing controllers that admit a convex parameterization using SLPs and SLCs is a strict superset of those that can be parameterized using quadratic invariance -- hence we provide a generalization of the QI framework, characterizing the broadest known class of constrained controllers that admit a convex parameterization;
\item formulate and analyze the SLS problem, which exploits SLPs and SLCs to define the broadest known class of constrained optimal control problems that can be solved using convex programming.  We show that the optimal control problems considered in the QI literature \cite{2012_Mahajan_Info_survey}, as well as the recently defined localized optimal control framework \cite{2015_Wang_H2} are all special cases of SLS problems.
\end{itemize}

\vspace{-4mm}
\subsection{Paper Structure} In Section \ref{sec:prelim}, we define the system model considered in this paper, and review relevant results from the distributed optimal control and QI literature. In Section \ref{sec:parameter} we define and analyze SLPs for state and output feedback problems, and provide a novel characterization of stable closed loop system responses and the controllers that achieve them -- the corresponding controller realization makes clear that SLCs imposed on the system responses carry over to the internal structure of the controller that achieves them. In Section \ref{sec:class}, we provide a catalog of SLCs that can be imposed on the system responses parameterized by the SLPs described in the previous section -- in particular, we show that by appropriately selecting these SLCs, we can provide convex characterizations of all stabilizing controllers satisfying QI subspace constraints, convex constraints on the Youla parameter, finite impulse response (FIR) constraints, sparsity constraints, spatiotemporal constraints \cite{2014_Wang_ACC, 2014_Wang_CDC,2014_Wang_Allerton,2015_Wang_H2, WMD_ACC17}, controller architecture constraints \cite{2015_Wang_Reg,Matni_RFD_TAC, Matni_RFD_CDC}, and any combination thereof. In Section \ref{sec:localizability}, we define and analyze the SLS problem, which incorporates SLPs and SLCs into an optimal control problem, and show that the distributed optimal control problem (\eqref{eq:decentralized} in Section \ref{sec:structured}) is a special case of SLS. \revsecond{We end with conclusions in Section \ref{sec:conclusion}.} 


 




\section{Preliminaries} \label{sec:prelim}
\subsection{System Model}
We consider discrete time linear time invariant (LTI) systems of the form
\begin{subequations} \label{eq:dynamics}
\begin{align}
x[t+1] &= A x[t] + B_1 w[t] + B_2 u[t] \label{eq:sys_x} \\
\bar{z}[t] &= C_1 x[t] + D_{11} w[t] + D_{12} u[t] \label{eq:sys_z} \\
y[t] &= C_2 x[t] + D_{21} w[t] + D_{22} u[t] \label{eq:sys_y}
\end{align}
\end{subequations}
where $x$, $u$, $w$, $y$, $\bar{z}$ are the state vector, control action, external disturbance, measurement, and regulated output, respectively. 
Equation \eqref{eq:dynamics} can be written in state space form as
\begin{equation}
\tf P = \left[ \begin{array}{c|cc} A & B_1 & B_2 \\ \hline C_1 & D_{11} & D_{12} \\ C_2 & D_{21} & D_{22} \end{array} \right] = \begin{bmatrix} \tf P_{11} & \tf P_{12} \\ \tf P_{21} & \tf P_{22} \end{bmatrix} \nonumber
\end{equation}
where $\tf P_{ij} = C_i(zI-A)^{-1}B_j + D_{ij}$. We refer to $\tf P$ as the open loop plant model. 

Consider a dynamic output feedback control law $\tf u = \tf K \tf y$. The controller $\tf K$ is assumed to have the state space realization
\begin{subequations} \label{eq:Kss}
\begin{align}
\xi[t+1] &= A_k \xi[t] + B_k y[t] \label{eq:Kss1} \\
u[t] &= C_k \xi[t] + D_k y[t], \label{eq:Kss2}
\end{align}
\end{subequations}
where $\xi$ is the internal state of the controller. We have $\tf K = C_k(zI-A_k)^{-1}B_k + D_k$. A schematic diagram of the interconnection of the plant $\tf P$ and the controller $\tf K$ is shown in Figure \ref{fig:pk}.

\begin{figure}[h]
      \centering
      \includegraphics[width=0.25\textwidth]{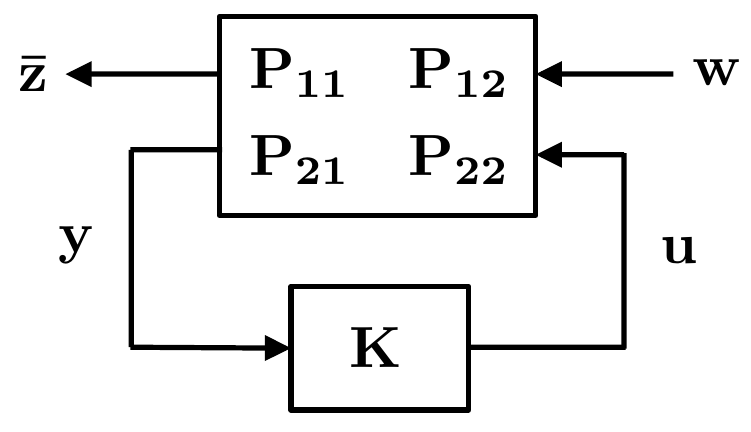}
      \caption{Interconnection of the plant $\tf P$ and controller $\tf K$.}
      \label{fig:pk}
\end{figure}

The following assumptions are made throughout the paper.
\begin{assumption}
The interconnection in Figure \ref{fig:pk} is well-posed -- the matrix $(I-D_{22}D_k)$ is invertible.
\end{assumption}
\begin{assumption}
Both the plant and the controller realizations are stabilizable and detectable; i.e., $(A,B_2)$ and $(A_k, B_k)$ are stabilizable, and $(A,C_2)$ and $(A_k,C_k)$ are detectable.
\end{assumption}

\rev{The goal of the optimal control problem is to find a controller $\tf K$ to stabilize the plant $\tf P$ and minimize a suitably chosen norm\footnote{\rev{Typical choices for the norm include $\mathcal{H}_2$ and $\mathcal{H}_\infty$.}} of the closed loop transfer matrix from external disturbance $\tf w$ to regulated output $ \tf{\bar{z}}$. This leads to the following centralized optimal control formulation:}
\rev{
\begin{eqnarray}
\underset{\tf K}{\text{minimize  }} &&|| \tf P_{11} + \tf P_{12} \tf K(I- \tf P_{22} \tf K)^{-1} \tf P_{21} || \nonumber\\
\text{subject to } && \tf K \text{ internally stabilizes } \tf P. \label{eq:centralized-add}
\end{eqnarray}
}
\vspace{-5mm}

\subsection{Youla Parameterization} \label{sec:move_Youla}
\revsecond{\rev{
A common technique to solve the optimal control problem \eqref{eq:centralized-add} is via the Youla parameterization, which is based on a doubly co-prime factorization of the plant, defined as follows.
\begin{definition}
A collection of stable transfer matrices, $\tf U_r$, $\tf V_r$, $\tf X_r$, $\tf Y_r$, $\tf U_l$, $\tf V_l$, $\tf X_l$, $\tf Y_l$ $\in \RHinf$ defines a doubly co-prime factorization of $\tf P_{22}$ if $\tf P_{22} = \tf V_r \tf U_r^{-1} = \tf U_l^{-1} \tf V_l$ and
\begin{equation}
\begin{bmatrix} \tf X_l & \tf{-Y}_l \\ \tf{-V}_l & \tf U_l \end{bmatrix} \begin{bmatrix} \tf U_r & \tf Y_r \\ \tf V_r & \tf X_r \end{bmatrix} = I. \nonumber
\end{equation}
\end{definition}
Such doubly co-prime factorizations can always be computed if $\tf P_{22}$ is stabilizable and detectable \cite{Zhou1996robust}.
}
\rev{
Let $\tf Q$ be the Youla parameter. From \cite{Zhou1996robust}, problem \eqref{eq:centralized-add} can be reformulated in terms of the Youla parameter as 
\begin{eqnarray}
\underset{\tf Q}{\text{minimize  }} &&|| \tf T_{11} + \tf T_{12} \tf Q \tf T_{21}|| \nonumber\\
\text{subject to } && \tf Q \in \RHinf \label{eq:trad_youla-add}
\end{eqnarray}
with $\tf T_{11} = \tf P_{11} + \tf P_{12} \tf Y_r \tf U_l \tf P_{21}$, $\tf T_{12} = - \tf P_{12} \tf U_r$, and $\tf T_{21} = \tf U_l \tf P_{21}$. 
}
\rev{
The benefit of optimizing over the Youla parameter $\tf Q$, rather than the controller $\tf K$, is that \eqref{eq:trad_youla-add} is convex with respect to the Youla parameter. One can then incorporate various convex design specifications \cite{Boyd_closed} in \eqref{eq:trad_youla-add} to customize the controller synthesis task.
Once the optimal Youla parameter $\tf Q$, or a suitable approximation thereof, is found in \eqref{eq:trad_youla-add}, we reconstruct the controller by setting $\tf K =  (\tf Y_r - \tf U_r \tf Q)(\tf X_r - \tf V_r \tf Q)^{-1}$. 
}
}

\subsection{Structured Controller Synthesis and QI} \label{sec:structured}
\rev{We now move our discussion to the distributed optimal control problem.}
We follow the paradigm adopted in \cite{2006_Rotkowitz_QI_TAC,2012_Lessard_two_player,2010_Shah_H2_poset,2013_Lamperski_H2,2013_Lessard_structure,2013_Scherer_Hinf,2014_Lessard_Hinf,2014_Matni_Hinf,2014_Tanaka_Triangular}, and focus on information asymmetry introduced by delays in the communication network -- this is a reasonable modeling assumption when one has dedicated physical communication channels (e.g., fiber optic channels), but may not be valid under wireless settings.  In the references cited above, locally acquired measurements are exchanged between sub-controllers subject to delays imposed by the communication network,\footnote{Note that this delay may range from 0, modeling instantaneous communication between sub-controllers, to infinite, modeling no communication between sub-controllers.} which manifest as subspace constraints on the controller itself.\rev{\footnote{\rev{For continuous time systems, the delays can be encoded via subspaces that may reside within $\mathcal{H}_\infty$ as opposed $\RHinf$.}}}

\rev{Let $\mathcal{C}$ be a subspace enforcing the information sharing constraints imposed on the controller $\tf K$. A distributed optimal control problem can then be formulated as \cite{2006_Rotkowitz_QI_TAC,2011_QIN, 2014_Lessard_convexity, 2014_Sabau_QI}:}
\begin{equation}
\begin{array}{rl}
\underset{\tf K}{\text{minimize  }} &\| \tf P_{11} + \tf P_{12} \tf K(I- \tf P_{22} \tf K)^{-1} \tf P_{21} \| \\
\text{subject to } & \tf K \text{ internally stabilizes } \tf P, \,\, \tf K \in \mathcal{C}.
\end{array}
\label{eq:decentralized}
\end{equation}

\rev{A summary of the main results from the distributed optimal control literature \cite{2006_Rotkowitz_QI_TAC,2012_Lessard_two_player,2010_Shah_H2_poset,2013_Lamperski_H2,2013_Lessard_structure,2013_Scherer_Hinf,2014_Lessard_Hinf,2014_Matni_Hinf,2014_Tanaka_Triangular} can be given as follows:} if the subspace $\mathcal{C}$ is quadratically invariant with respect to $\tf P_{22}$ \revsecond{($\tf K \tf P_{22} \tf K \in \mathcal{C}, \,\, \forall \tf K \in \mathcal{C}$)} \cite{2006_Rotkowitz_QI_TAC}, then the set of all stabilizing controllers lying in subspace $\mathcal{C}$ can be parameterized by those stable transfer matrices $\tf Q \in \RHinf$ satisfying $\mathfrak{M}(\tf Q)\in \mathcal{C}$, \rev{for $\mathfrak{M}(\tf Q) := \tf K(I- \tf P_{22} \tf K)^{-1} = (\tf Y_r - \tf U_r \tf Q) \tf U_l$.\footnote{\rev{By definition, we have $\tf P_{22} = \tf V_r \tf U_r^{-1} = \tf U_l^{-1} \tf V_l$. This implies that the transfer matrices $\tf U_r$ and $\tf U_l$ are both invertible. Therefore, $\mathfrak{M}$ is an invertible affine map of the Youla parameter $\tf Q$.}}}
Further, these conditions can be viewed as tight, in the sense that quadratic invariance is also a necessary condition \cite{2011_QIN, 2014_Lessard_convexity} for a subspace constraint $\mathcal{C}$ on the controller $\tf K$ to be enforced via a convex constraint on the Youla parameter $\tf Q$.

This allows the optimal control problem \eqref{eq:decentralized} to be recast as the following convex model matching problem:
\begin{equation} 
\begin{array}{rl}
\underset{\tf Q}{\text{minimize}} & \| \tf T_{11} + \tf T_{12} \tf Q \tf T_{21}\| \\
\text{subject to} & \tf Q \in \RHinf, \,\, \mathfrak{M}(\tf Q) \in \mathcal{C}.
\end{array}
\label{eq:trad_youla2}
\end{equation}

\subsection{\revsecond{QI imposes limitations on controller sparsity}} \label{sec:limitation}
\revsecond{When working with large-scale systems, it is natural to impose that sub-controllers only collect information from a local subset of all other sub-controllers. This can be enforced by setting the subspace constraint $\mathcal{C}$ in problem \eqref{eq:decentralized} to encode a suitable sparsity pattern $\tf K_{ij} = 0$,\footnote{$\tf K_{ij}$ denotes the $(i,j)$-entry of the transfer matrix $\tf K$.} for some $i, j$.  However, if the plant $\tf P_{22}$ is dense (i.e., if the underlying system is strongly connected), which may occur even if the system matrices $(A,B_2,C_2)$ are sparse, then \emph{any} such sparsity constraint is not quadratically invariant with respect to the plant $\tf P_{22}$: this follows immediately from the algebraic definition of QI $\tf K \tf P_{22} \tf K \in \mathcal{C}, \,\, \forall \tf K \in \mathcal{C}$.  As QI is a necessary and sufficient condition for the subspace constraint $\tf K \in \mathcal{C}$ to be enforced via a convex constraint on the Youla parameter $\tf Q$, we conclude that for strongly connected systems, any sparsity constraint imposed on the controller $\tf K$ can only be enforced via a non-convex constraint on Youla parameter.  A major motivation for the SLA developed in this paper was to circumvent this limitation of the QI framework -- we revisit this discussion in Section \ref{sec:beyond}, and show, through the use of a simple example, that the SLA does indeed allow for these limitations to be overcome.}

\section{System Level Parameterization} \label{sec:parameter}
\revsecond{In this section, we propose a novel parameterization of internally stabilizing controllers centered around \emph{system responses}, which are defined by the closed loop maps from process and measurement disturbances to state and control action. We show} that for a given system, the set of stable \revsecond{closed loop system responses that are achievable by an internally stabilizing LTI controller} is an affine subspace of $\RHinf$, and that the corresponding internally stabilizing controller achieving the desired system response admits a particularly simple and transparent realization.

We begin by analyzing the state feedback case, as it has a simpler characterization and allows us to provide intuition about the construction of a controller that achieves a desired system response. With this intuition in hand, we present our results for the output feedback setting, which is the main focus of this paper. \revsecond{We conclude the section with a comparison of the pros and cons of using the SL and Youla parameterizations.}

%
%
%
%

\subsection{State Feedback} \label{sec:sf}
\revsecond{Consider a state feedback model given by}
\begin{equation}
\tf P = \left[ \begin{array}{c|cc} A & B_1 & B_2 \\ \hline C_1 & D_{11} & D_{12} \\ I & 0 & 0 \end{array} \right]. \label{eq:sfplant}
\end{equation}
The $z$-transform of the state dynamics \eqref{eq:sys_x} is given by
\begin{equation}
(zI - A) \tf x = B_2 \tf u + \ttf{\delta_x}, \label{eq:zsfb}
\end{equation}
where we let $ \ttf{\delta_x} := B_1\tf w$ denote the disturbance affecting the state.
We define $\tf R$ to be the system response mapping the external disturbance $ \ttf{\delta_x}$ to the state $\tf x$, and $\tf M$ to be the system response mapping the disturbance $\ttf{\delta_x}$ to the control action $\tf u$. \revsecond{For a given dynamic state feedback control rule $\tf u = \tf K \tf x$ into \eqref{eq:zsfb}, we define the system response $\{\tf R, \tf M\}$ achieved by the controller $\tf K$ to be
\begin{eqnarray}
\tf R &=& (zI-A-B_2 \tf K)^{-1} \nonumber\\
\tf M &=& \tf K (zI-A-B_2 \tf K)^{-1}, \label{eq:KRM}
\end{eqnarray}
from which it follows that $\tf x = \tf R \ttf{\delta_x}$ and $\tf u = \tf M \ttf{\delta_x}$.}

\revsecond{
Similarly, given transfer matrices $\{ \tf R, \tf M\}$, we say that they define an achievable system response for the system \eqref{eq:sfplant} if there exists a LTI controller $\tf K$ such that $\tf x = \tf R \ttf{\delta_x}$ and $\tf u = \tf M \ttf{\delta_x}$, for $\{\tf R, \tf M\}$ as defined in equation \eqref{eq:KRM}.}

The main result of this subsection is an algebraic characterization of the set $\{\tf R, \tf M\}$ of state-feedback system responses that are achievable by an internally stabilizing controller $\tf K$, as stated in the following theorem.

\begin{theorem}
For the state feedback system \eqref{eq:sfplant}, the following are true:
\begin{enumerate}[(a)]
  \item The affine subspace defined by
  \begin{subequations} \label{eq:state_fb}
\begin{align}
& \begin{bmatrix} zI - A & -B_2 \end{bmatrix} \begin{bmatrix} \tf R \\ \tf M \end{bmatrix} = I \label{eq:state_fb1}\\
& \tf R, \tf M \in \frac{1}{z} \mathcal{RH}_\infty \label{eq:state_fb2}
\end{align}
\end{subequations}
 parameterizes all system responses \eqref{eq:KRM} achievable by an internally stabilizing state feedback controller $\tf K$. \label{thm1-4}
  \item For any transfer matrices  $\{\tf R,\tf M\}$ satisfying \eqref{eq:state_fb}, the controller $\tf K = \tf M \tf R^{-1}$ achieves the desired system response \eqref{eq:KRM}.\rev{\footnote{\rev{Note that for any transfer matrices $\{\tf R,\tf M\}$ satisfying \eqref{eq:state_fb}, the transfer matrix $\tf R$ is always invertible because its leading spectral component $\frac{1}{z} I$ is invertible.  This is also true for the transfer matrices defined in equation \eqref{eq:KRM}.}}}. \revsecond{Further, if the controller $\tf K = \tf M \tf R^{-1}$ is implemented as in Fig. \ref{fig:sf}, then it is internally stabilizing. }
  \label{thm1-3} 
  \end{enumerate} \label{thm:sf}
\end{theorem}

The rest of this subsection is devoted to the proof of Theorem \ref{thm:sf}.  

\subsubsection*{Necessity}  \rev{The necessity of a stable and achievable system response $\{\tf R, \tf M\}$ lying in the affine subspace \eqref{eq:state_fb} is shown in the following lemma.}
\rev{\begin{lemma}[Necessity of conditions \eqref{eq:state_fb}] \label{lem:sf_nec}
Consider the state feedback system \eqref{eq:sfplant}. Let $\{\tf R, \tf M\}$ be the system response achieved by an internally stabilizing controller $\tf K$. Then, $\{\tf R, \tf M\}$ is a solution of \eqref{eq:state_fb}.
\end{lemma} 
\begin{proof}
\revsecond{Equation \eqref{eq:state_fb1} follows directly from \eqref{eq:zsfb}, which holds for the system response achieved by any controller. For an internally stabilizing controller, the system response $\{\tf R, \tf M\}$ is in $\RHinf$ by definition of internal stability. From \eqref{eq:KRM} and the properness of $\tf K$, the system response is strictly proper, implying equation \eqref{eq:state_fb2} and completing the proof.}
\end{proof}}


\begin{remark}
We show in Lemma \ref{lemma:1} in Appendix \ref{sec:proof} that the feasibility of \eqref{eq:state_fb} is equivalent to the stabilizability of the pair $(A, B_2)$. In this sense, the conditions described in \eqref{eq:state_fb} provides an alternative definition of the stabilizability of a system. A dual argument is also provided to characterize the detectability of the pair $(A,C_2)$.
\end{remark}

\subsubsection*{Sufficiency}  Here we show that for any system response $\{\tf R, \tf M\}$ lying in the affine subspace \eqref{eq:state_fb}, we can construct an internally stabilizing controller $\tf K$ that leads to the desired system response \eqref{eq:KRM}.  

\begin{figure}[ht!]
      \centering
      \includegraphics[width=0.4\textwidth]{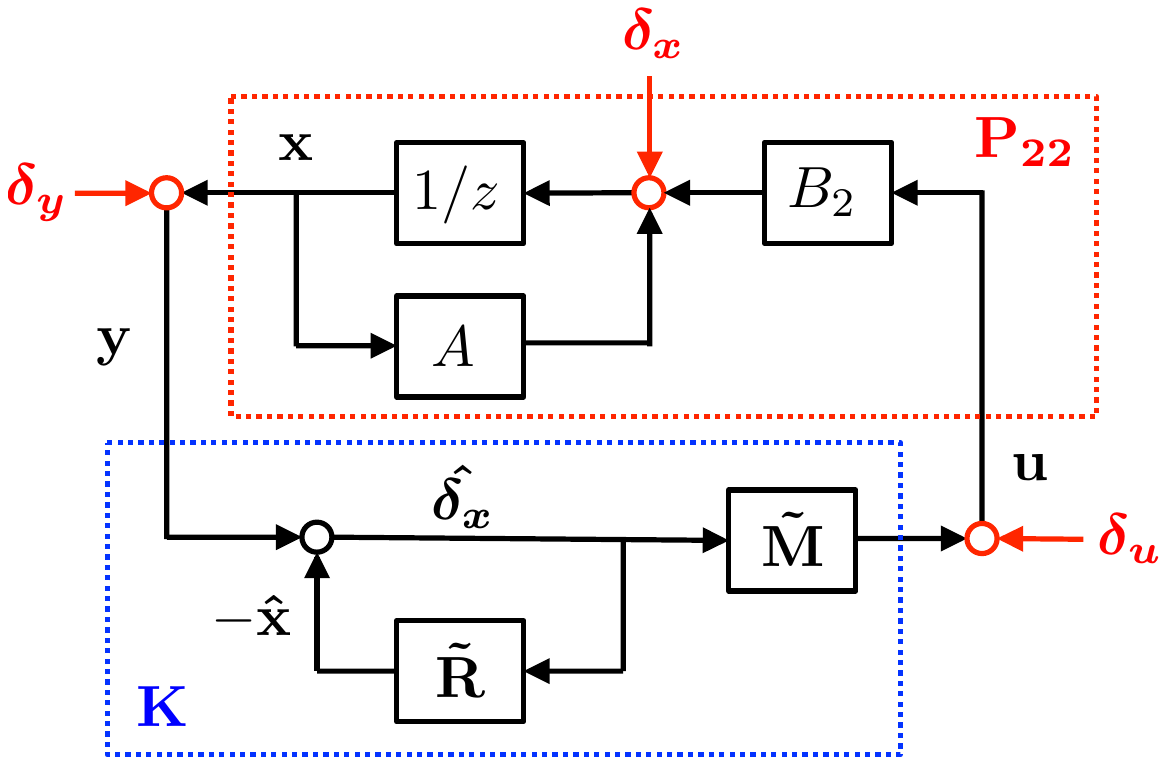}
      \caption{The proposed state feedback controller structure, with $\tf{\tilde{R}} = I - z \tf R$ and $\tf{\tilde{M}} = z \tf M$.}
      \label{fig:sf}
\end{figure}

%

\revsecond{Consider the block diagram shown in Figure \ref{fig:sf}, where here $\tf{\tilde{R}} = I - z \tf R$ and $\tf{\tilde{M}} = z \tf M$.} It can be checked that $z\tf{\tilde{R}},\tf{\tilde{M}} \in \RHinf$, and hence the internal feedback loop between $\ttf{\hat{\delta}_x}$ and the reference state trajectory $\tf{\hat{x}}$ is well defined.  As is standard, we introduce external perturbations $\delta_x, \delta_y$, and $\delta_u$ into the system and note that the perturbations entering other links of the block diagram can be expressed as a combination of $(\ttf{\delta_x}, \ttf{\delta_y}, \ttf{\delta_u})$ being acted upon by some stable transfer matrices.\footnote{The matrix $A$ may define an unstable system, but viewed as an element of $\mathcal{F}_0$, defines a stable (FIR) transfer matrix.}  Hence the standard definition of internal stability applies, and we can use a bounded-input bounded-output argument (e.g., Lemma $5.3$ in \cite{Zhou1996robust}) to conclude that it suffices to check the stability of the nine closed loop transfer matrices from perturbations $(\ttf{\delta_x}, \ttf{\delta_y}, \ttf{\delta_u})$ to the internal variables $(\tf x, \tf u, \ttf{\hat{\delta_x}})$ to determine the internal stability of the structure as a whole. \rev{With this in mind, we can prove the sufficiency of Theorem \ref{thm:sf} via the following lemma.} 


%

\begin{lemma}[Sufficiency of conditions \eqref{eq:state_fb}] \label{lem:sf_suf}
Consider the state feedback system \eqref{eq:sfplant}.
Given any system response $\{\tf R, \tf M\}$ lying in the affine subspace described by \eqref{eq:state_fb}, the state feedback controller $\tf K = \tf M \tf R^{-1}$, with structure shown in Figure \ref{fig:sf}, internally stabilizes the plant. In addition, the desired system response, as specified by $\tf x = \tf R  \ttf{\delta_x}$ and $\tf u = \tf M  \ttf{\delta_x}$, is achieved. \label{thm1-2}
\end{lemma}
\begin{proof}
We first note that from Figure \ref{fig:sf}, we can express the state feedback controller $\tf K$ as $\tf K = \tf{\tilde{M}} (I - \tf{\tilde{R}})^{-1} = (z \tf M)(z \tf R)^{-1} = \tf M \tf R^{-1}$. 
Now, for any system response $\{\tf R, \tf M\}$ lying in the affine subspace described by \eqref{eq:state_fb}, we construct a controller using the structure given in Figure \ref{fig:sf}. To show that the constructed controller internally stabilizes the plant, we list the following equations from Figure \ref{fig:sf}:
\begin{equation*}
\begin{array}{rcl}
z \tf x &=& A \tf x + B_2 \tf u + \ttf{\delta_x}\\
\tf u &=& \tf{\tilde{M}} \ttf{\hat{\delta_x}} + \ttf{\delta_u} \\
\ttf{\hat{\delta}_x} &=& \tf x + \ttf{\delta_y} + \tf{\tilde{R}} \ttf{\hat{\delta}_x}. 
\end{array}
\end{equation*}
Routine calculations show that the closed loop transfer matrices from $(\ttf{\delta_x}, \ttf{\delta_y}, \ttf{\delta_u})$ to $(\tf x, \tf u, \ttf{\hat{\delta}_x})$ are given by
\begin{equation}
\begin{bmatrix} \tf x \\ \tf u \\ \ttf{\hat{\delta}_x} \end{bmatrix} = \begin{bmatrix} \tf R & \tf{-\tilde{R}} - \tf R A & \tf R B_2 \\ \tf M & \tf{\tilde{M}} - \tf M A & I + \tf M B_2 \\ \frac{1}{z}I & I-\frac{1}{z}A & \frac{1}{z}B_2 \end{bmatrix} \begin{bmatrix} \ttf{\delta_x} \\ \ttf{\delta_y} \\ \ttf{\delta_u} \end{bmatrix}. \label{eq:sf_cl}
\end{equation}
As all nine transfer matrices in \eqref{eq:sf_cl} are stable, the implementation in Figure \ref{fig:sf} is internally stable. Furthermore, the desired system response $\{\tf R, \tf M\}$, from $\ttf{\delta_x}$ to $(\tf x, \tf u)$, is achieved.
\end{proof}

\revsecond{
\begin{remark}
The controller parameterization $\tf K = \tf M \tf R^{-1}$ can also be derived by rewriting \eqref{eq:state_fb} as
\begin{equation}
\begin{bmatrix} I - \frac{1}{z} A & -\frac{1}{z} B_2 \end{bmatrix} \begin{bmatrix} z \tf R \\ z \tf M \end{bmatrix} = I, \,\, 
z \tf R, z \tf M \in \RHinf. \nonumber
\end{equation}
Note that $\begin{bmatrix} I - \frac{1}{z} A & -\frac{1}{z} B_2 \end{bmatrix}$ is a left coprime factorization of the plant model.  Classical methods therefore allow for the controller $\tf K =  (z \tf M)(z \tf R)^{-1} = \tf M \tf R^{-1}$ to be obtained via the Youla parameterization. Although the controller can be implemented via the dynamic feedback gain $\tf K = \tf M \tf R^{-1}$, we show in Section \ref{sec:class} that the proposed realization in Figure \ref{fig:sf} has significant advantages. Specifically, this implementation allows us to connect constraints imposed on the system response to constraints on the internal blocks of the controller implementation. 
\end{remark}
}

\subsubsection*{Summary}
Theorem \ref{thm:sf} provides a necessary and sufficient condition for the system response $\{\tf R, \tf M\}$ to be stable and achievable, in that elements of the affine subspace defined by \eqref{eq:state_fb} parameterize all stable system responses achievable via state-feedback, as well as the internally stabilizing controllers that achieve them.  \rev{Further, Figure \ref{fig:sf} provides an internally stabilizing realization for a controller achieving the desired response.}

\subsection{Output Feedback with $D_{22} = 0$} \label{sec:of}

We now extend the arguments of the previous subsection to the output feedback setting, and begin by considering the case of a strictly proper plant
\begin{equation}
\tf P = \left[ \begin{array}{c|cc} A & B_1 & B_2 \\ \hline C_1 & D_{11} & D_{12} \\ C_2 & D_{21} & 0 \end{array} \right]. \label{eq:ofplant}
\end{equation}

Letting $\delta_x [t] = B_1 w[t]$ denote the disturbance on the state, and $\delta_y [t] = D_{21} w[t]$ denote the disturbance on the measurement, the dynamics defined by plant \eqref{eq:ofplant} can be written as
\begin{eqnarray}
x[t+1] &=& A x[t] + B_2 u[t] + \delta_x [t] \nonumber \\
y[t] &=& C_2 x[t] + \delta_y [t]. \label{eq:sys_out}
\end{eqnarray}

Analogous to the state-feedback case, we define a system response $\{\tf R, \tf M, \tf N, \tf L\}$ from perturbations $(\ttf{\delta_x}, \ttf{\delta_y})$ to state and control inputs $(\tf x,\tf u)$ via the following relation:
\begin{equation}
\begin{bmatrix} \tf x \\ \tf u \end{bmatrix} = \begin{bmatrix} \tf R & \tf N \\ \tf M & \tf L \end{bmatrix} \begin{bmatrix} \ttf{\delta_x} \\ \ttf{\delta_y} \end{bmatrix}. \label{eq:cltm_df}
\end{equation}

Substituting the output feedback control law $\tf u = \tf K \tf y$ into the z-transform of system equation \eqref{eq:sys_out}, we obtain
\begin{equation}
(zI - A - B_2 \tf K C_2) \tf x = \ttf{\delta_x} + B_2 \tf K \ttf{\delta_y}. \nonumber
\end{equation}
For a proper controller $\tf K$, the transfer matrix $(zI-A-B_2 \tf K C_2)$ is always invertible, hence we obtain the following \rev{equivalent} expressions for the system response \eqref{eq:cltm_df} in terms of an output feedback controller $\tf K$:
\begin{align}
\tf R &= (zI - A - B_2 \tf K C_2)^{-1} \nonumber\\
\tf M &= \tf K C_2 \tf R \nonumber\\
\tf N &= \tf R B_2 \tf K \nonumber\\
\tf L  &= \tf K + \tf K C_2 \tf R B_2 \tf K. \label{eq:K_relation}
\end{align}

We now present one of the main results of the paper: an algebraic characterization of the set $\{\tf R, \tf M, \tf N, \tf L\}$ of output-feedback system responses that are achievable by an internally stabilizing controller $\tf K$.

\begin{theorem}
For the output feedback system \eqref{eq:ofplant}, the following are true:
\begin{enumerate}[(a)]
  \item The affine subspace described by:
  \begin{subequations} \label{eq:output_fb}
\begin{align}
\begin{bmatrix} zI - A & -B_2 \end{bmatrix}
\begin{bmatrix} \tf R & \tf N \\ \tf M & \tf L \end{bmatrix} &= 
\begin{bmatrix} I & 0 \end{bmatrix} \label{eq:output_fb1}\\
\begin{bmatrix} \tf R & \tf N \\ \tf M & \tf L \end{bmatrix}
 \begin{bmatrix} zI - A \\ -C_2 \end{bmatrix} &= 
\begin{bmatrix} I \\ 0 \end{bmatrix} \label{eq:output_fb2} \\
\tf R, \tf M, \tf N \in \frac{1}{z} \mathcal{RH}_\infty, \quad & \tf L \in \mathcal{RH}_\infty \label{eq:output_fb3}
\end{align}
\end{subequations} 
 parameterizes all system responses \eqref{eq:K_relation} achievable by an internally stabilizing controller $\tf K$. \label{thm2-4}
  \item For any transfer matrices $\{\tf R, \tf M, \tf N, \tf L\}$ satisfying \eqref{eq:output_fb}, the controller $\tf K = \tf L - \tf M \tf R^{-1} \tf N$ achieves the desired response \eqref{eq:K_relation}.\rev{\footnote{\rev{Note that for any transfer matrices $\{\tf R, \tf M, \tf N, \tf L\}$ satisfying \eqref{eq:output_fb}, the transfer matrix $\tf R$ is always invertible because its leading spectral component $\frac{1}{z} I$ is invertible.  The same holds true for the transfer matrices defined in equation \eqref{eq:K_relation}.}}} \revsecond{Further, if the controller is implemented as in Fig. \ref{fig:of}, then it is internally stabilizing.} \label{thm2-3} 
  \end{enumerate} \label{thm:of}
\end{theorem}

\subsubsection*{Necessity}  \rev{The necessity of a stable and achievable system response $\{\tf R, \tf M, \tf N, \tf L \}$ lying in the affine subspace \eqref{eq:output_fb} is shown in the following lemma.}

\begin{lemma}[Necessity of conditions \eqref{eq:output_fb}] \label{lem:of_nec}
Consider the output feedback system \eqref{eq:ofplant}. Let $\{ \tf R, \tf M, \tf N, \tf L\}$, with $\tf x = \tf R \ttf{\delta_x} + \tf N \ttf{\delta_y}$ and  $\tf u = \tf M \ttf{\delta_x} + \tf L \ttf{\delta_y}$, be the system response achieved by an internally stabilizing control law $\tf u = \tf K \tf y$. Then $\{\tf R, \tf M, \tf N, \tf L\}$ lies in the affine subspace described by \eqref{eq:output_fb}.
\end{lemma}
\begin{proof}
Consider an internally stabilizing controller $\tf K$ with state space realization \eqref{eq:Kss}.
Combining \eqref{eq:Kss} with the system equation \eqref{eq:sys_out}, we obtain the closed loop dynamics
\begin{equation}
\begin{bmatrix} z \tf x \\ z \ttf{\xi} \end{bmatrix} = \begin{bmatrix} A + B_2 D_k C_2 & B_2 C_k \\ B_k C_2 & A_k \end{bmatrix} \begin{bmatrix} \tf x \\ \ttf{\xi} \end{bmatrix} + \begin{bmatrix} I & B_2 D_k \\ 0 & B_k \end{bmatrix} \begin{bmatrix} \ttf{\delta_x} \\ \ttf{\delta_y} \end{bmatrix}. \nonumber
\end{equation}
From the assumption that $\tf K$ is internally stabilizing, we know that the state matrix of the above equation is a stable matrix (Lemma $5.2$ in \cite{Zhou1996robust}). The system response achieved by $\tf u = \tf K \tf y$ is given by 
\begin{equation}
\begin{bmatrix} \tf R & \tf N \\ \tf M & \tf L \end{bmatrix} = \left[ \begin{array}{cc | cc} A + B_2 D_k C_2 & B_2 C_k & I & B_2 D_k \\ B_k C_2 & A_k & 0 & B_k \\ \hline I & 0 & 0 & 0 \\ D_k C_2 & C_k & 0 & D_k \end{array} \right], \label{eq:ss_relation}
\end{equation}
which satisfies \eqref{eq:output_fb3}.
In addition, it can be shown by routine calculation that \eqref{eq:ss_relation} satisfies both \eqref{eq:output_fb1} and \eqref{eq:output_fb2} for arbitrary $(A_k, B_k, C_k, D_k)$. This completes the proof.
\end{proof}


\rev{
\begin{remark}
We show in Lemma \ref{lem:stab_det} of Appendix \ref{sec:proof} that the feasibility of \eqref{eq:output_fb} is equivalent to the stabilizability and detectability of the triple $(A, B_2, C_2)$. In this sense, the conditions described in \eqref{eq:output_fb} provide an alternative definition of stabilizability and detectability.
\end{remark}
}



\subsubsection*{Sufficiency}

\begin{figure}[ht!]
      \centering
      \includegraphics[width=0.48\textwidth]{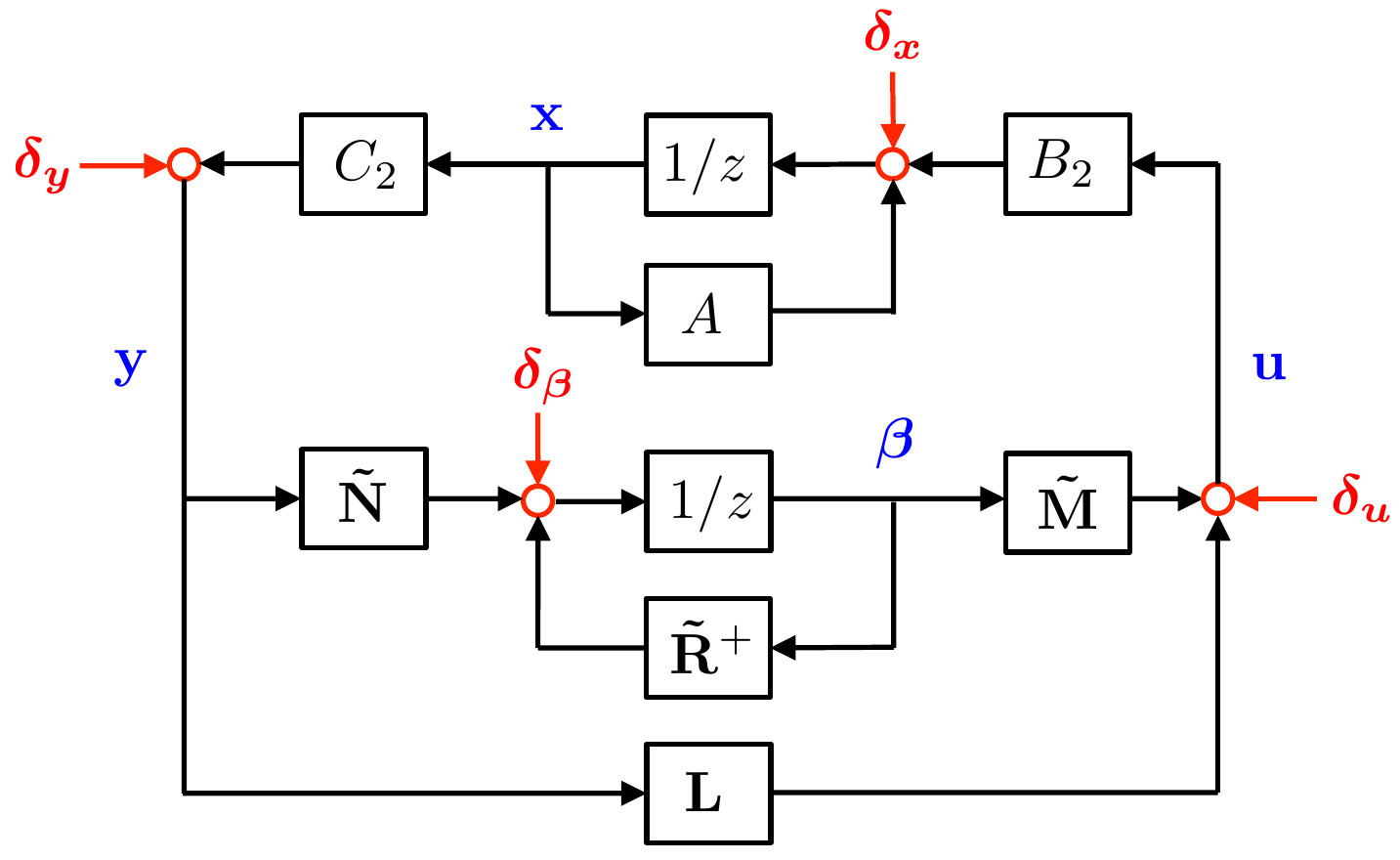}
      \caption{The proposed output feedback controller structure, with $\tf{\tilde{R}^+} = z \tf{\tilde{R}} = z(I - z \tf R)$, $\tf{\tilde{M}} = z \tf M$, and $\tf{\tilde{N}} = -z \tf N$.}
      \label{fig:of}
\end{figure}

Here we show that for any system response $\{\tf R, \tf M, \tf N, \tf L\}$ lying in the affine subspace \eqref{eq:output_fb}, there exists an internally stabilizing controller $\tf K$ that leads to the desired system response \eqref{eq:K_relation}. From the relations in \eqref{eq:K_relation}, we notice the identity $\tf K = \tf L - \tf K C_2 \tf R B_2 \tf K = \tf L - \tf M \tf R^{-1} \tf N$. This relation leads to the controller structure given in Figure \ref{fig:of}, with $\tf{\tilde{R}^+} = z\tf{\tilde{R}} = z (I - z \tf R)$, $\tf{\tilde{M}} = z \tf M$, and $\tf{\tilde{N}} = -z \tf N$. As was the case for the state feedback setting, it can be verified that $\tf{\tilde{R}^{+}}, \tf{\tilde{M}}$, and $\tf{\tilde{N}}$ are all in $\RHinf$. Therefore, the structure given in Figure \ref{fig:of} is well defined. In addition, all of the blocks in Figure \ref{fig:of} are stable filters -- thus, as long as the origin $(x,\beta) = (0,0)$ is asymptotically stable, all signals internal to the block diagram will decay to zero. To check the internal stability of the structure, we introduce external perturbations $\ttf{\delta_x}, \ttf{\delta_y}$, $\ttf{\delta_u}$, and $\ttf{\delta_\beta}$ to the system. The perturbations appearing on other links of the block diagram can all be expressed as a combination of the perturbations $(\ttf{\delta_x}, \ttf{\delta_y}, \ttf{\delta_u}, \ttf{\delta_\beta})$ being acted upon by some stable transfer matrices, and so it suffices to check the input-output stability of the closed loop transfer matrices from perturbations $(\ttf{\delta_x}, \ttf{\delta_y}, \ttf{\delta_u}, \ttf{\delta_\beta})$ to controller signals $(\tf x, \tf u, \tf y, \ttf{\beta})$ to determine the internal stability of the structure \cite{Zhou1996robust}. \rev{With this in mind, we can prove the sufficiency of Theorem \ref{thm:of} via the following lemma.}

%

\begin{lemma}[Sufficiency of conditions \eqref{eq:output_fb}]\label{lem:of_suf}
Consider the output feedback system \eqref{eq:ofplant}.
For any system response $\{\tf R, \tf M, \tf N, \tf L\}$ lying in the affine subspace defined by \eqref{eq:output_fb}, the controller $\tf K = \tf L - \tf M \tf R^{-1} \tf N$ (with structure shown in Figure \ref{fig:of}) internally stabilizes the plant. In addition, the desired system response, as specified by $\tf x = \tf R \ttf{\delta_x} + \tf N \ttf{\delta_y}$ and $\tf u = \tf M \ttf{\delta_x} + \tf L \ttf{\delta_y}$, is achieved. 
\end{lemma}
\begin{proof}
For any system response $\{\tf R, \tf M, \tf N, \tf L\}$ lying in the affine subspace defined by \eqref{eq:output_fb}, we construct a controller using the structure given in Figure \ref{fig:of}. 
We now check the stability of the closed loop transfer matrices from the perturbations $(\ttf{\delta_x}, \ttf{\delta_y}, \ttf{\delta_u}, \ttf{\delta_\beta})$ to the internal variables $(\tf x, \tf u, \tf y, \ttf{\beta})$. We have the following equations from Figure \ref{fig:of}:
\begin{eqnarray}
z \tf x &=& A \tf x + B_2 \tf u + \ttf{\delta_x} \nonumber\\
\tf y &=& C_2 \tf x + \ttf{\delta_y} \nonumber\\
z \ttf{\beta} &=& \tf{\tilde{R}^+} \ttf{\beta} + \tf{\tilde{N}} \tf y + \ttf{\delta_\beta} \nonumber\\
\tf u &=& \tf{\tilde{M}} \ttf{\beta} + \tf L \tf y + \ttf{\delta_u}. \nonumber
\end{eqnarray}
Combining these equations with the relations in \eqref{eq:output_fb1} - \eqref{eq:output_fb2}, we summarize the closed loop transfer matrices from $(\ttf{\delta_x}, \ttf{\delta_y}, \ttf{\delta_u}, \ttf{\delta_\beta})$ to $(\tf x, \tf u, \tf y, \ttf{\beta})$ in Table \ref{Table:1}.

\begin{table}[t!]
 \caption{Closed Loop Maps from Perturbations to Internal Variables}
 \label{Table:1}
\begin{center}
\renewcommand{\arraystretch}{2}
    \begin{tabular}{| c | c | c | c | c | c |}
    \hline 
    & $\ttf{\delta_x}$ & $\ttf{\delta_y}$ & $\ttf{\delta_u}$ & $\ttf{\delta_\beta}$ \\ \hline
    $\tf x$ & $\tf R$ & $\tf N$ & $\tf R B_2$ & $\frac{1}{z} \tf N C_2$ \\ \hline
    $\tf u$ & $\tf M$ & $\tf L$ & $I + \tf M B_2$ & $\frac{1}{z} \tf L C_2$ \\ \hline
    $\tf y$ & $C_2 \tf R$ & $I + C_2 \tf N$ & $C_2 \tf R B_2$ & $\frac{1}{z} C_2 \tf N C_2$ \\ \hline
    $\ttf{\beta}$ & $-\frac{1}{z} B_2 \tf M$ & $-\frac{1}{z} B_2 \tf L$ & $-\frac{1}{z} B_2 \tf M B_2$ & $\frac{1}{z} I - \frac{1}{z^2} (A + B_2 \tf L C_2)$ \\ \hline
    \end{tabular}
\end{center}
\end{table}

Equation \eqref{eq:output_fb3} implies that all sixteen transfer matrices in Table \ref{Table:1} are stable, so the implementation in Figure \ref{fig:of} is internally stable. Furthermore, the desired system response from $(\ttf{\delta_x}, \ttf{\delta_y})$ to $(\tf x,\tf u)$ is achieved.
\end{proof}

The controller implementation of Figure \ref{fig:of} is governed by the following equations:
\begin{eqnarray}
z \ttf{\beta} &=& \tf{\tilde{R}^{+}} \ttf{\beta} + \tf{\tilde{N}} \tf y \nonumber\\
\tf u &=& \tf{\tilde{M}} \ttf{\beta} + \tf L \tf y, \label{eq:ss_like}
\end{eqnarray}
which can be informally interpreted as an extension of the state-space realization \eqref{eq:Kss} of a controller $\tf K$. In particular, the realization equations \eqref{eq:ss_like} can be viewed as a state-space like implementation where the constant matrices $A_K, B_K, C_K, D_K$ of the state-space realization \eqref{eq:Kss} are replaced with stable proper transfer matrices $\tf{\tilde{R}^{+}}, \tf{\tilde{M}}, \tf{\tilde{N}}, \tf L$.  The benefit of this implementation is that arbitrary convex constraints imposed on the transfer matrices $\tf{\tilde{R}^{+}}, \tf{\tilde{M}}, \tf{\tilde{N}}, \tf L$ carry over directly to the controller implementation.  We show in Section \ref{sec:class} that this allows for a class of structural (locality) constraints to be imposed on the system response (and hence the controller) that are crucial for extending controller synthesis methods to large-scale systems. \rev{In contrast, we recall that imposing general convex constraints on the controller $\tf K$ or directly on its state-space realization $A_K, B_K, C_K, D_K$ do not lead to convex optimal control problems.}

\revsecond{\begin{remark}
The controller implementation \eqref{eq:ss_like} admits the following equivalent representation
\begin{equation}
\begin{bmatrix} \tf R & \tf N \\ \tf M & \tf L \end{bmatrix} \begin{bmatrix}z \ttf{\beta} \\ \tf y \end{bmatrix} = \begin{bmatrix} 0 \\ \tf u\end{bmatrix}, \label{eq:ros_rep}
\end{equation}
allowing for an interesting interpretation of the controller $\tf K = \tf L - \tf M \tf R^{-1} \tf N$ in terms of Rosenbrock system matrix representations \cite{Rosenbrock}. In particular, the system response \eqref{eq:cltm_df} specifies a Rosenbrock system matrix representation of the controller that achieves it.
\end{remark}}

\subsubsection*{Summary}
Theorem \ref{thm:of} provides a necessary and sufficient condition for the system response $\{\tf R, \tf M, \tf N, \tf L\}$ to be stable and achievable, in that elements of the affine subspace defined by \eqref{eq:output_fb} parameterize all stable achievable system responses, as well as all internally stabilizing controllers that achieve them.  \rev{Further, Figure \ref{fig:of} provides an internally stabilizing realization for a controller achieving the desired response.}

\rev{\subsection{Specialized Implementations for Open-loop Stable Systems}}
\rev{In this subsection, we propose two specializations of the controller implementation in Figure \ref{fig:of} for open loop stable systems.} From Table \ref{Table:1}, if we set $\ttf{\delta_u}$ and $\ttf{\delta_\beta}$ to $0$, it follows that $\ttf{\beta} = -\frac{1}{z}B_2 \tf u$. 
This leads to a simpler controller implementation given by $\tf u = \tf L \tf y - \tf M B_2 \tf u$, with the corresponding controller structure shown in Figure \ref{fig:alt1}. 
This implementation can also be obtained from the identity $ \tf K = (I + \tf M B_2)^{-1} \tf L$, which follows from the relations in \eqref{eq:K_relation}. 
Unfortunately, as shown below, this implementation is internally stable only when the open loop plant is stable.  

For the controller implementation and structure shown in Figure \ref{fig:alt1}, the closed loop transfer matrices from perturbations to the internal variables are given by
\begin{equation}
\begin{bmatrix} \tf x \\ \tf u \end{bmatrix} = \begin{bmatrix} \tf R & \tf N & \tf R B_2 & (zI-A)^{-1} B_2 \\ \tf M & \tf L & I + \tf M B_2 & I \end{bmatrix} \begin{bmatrix} \ttf{\delta_x} \\ \ttf{\delta_y} \\ \ttf{\delta_u} \\ \ttf{\delta_\beta} \end{bmatrix}. \label{eq:alt1}
\end{equation}
When $A$ defines a stable system, the implementation in Figure \ref{fig:alt1} is internally stable. However, when the open loop plant is unstable (and the realization $(A,B_2)$ is stabilizable), the transfer matrix $(zI-A)^{-1} B_2$ is unstable. From \eqref{eq:alt1}, the effect of the perturbation $\ttf{\delta_\beta}$ can lead to instability of the closed loop system. This structure thus shows the necessity of introducing and analyzing the effects of perturbations $\ttf{\delta_\beta}$ on the controller internal state.

Alternatively, if we start with the identity $\tf K = \tf L (I + C_2 \tf N)^{-1}$, which also follows from \eqref{eq:K_relation}, we obtain the controller structure shown in Figure \ref{fig:alt2}. The closed loop map from perturbations to internal signals is then given by
\begin{equation}
\begin{bmatrix} \tf x \\ \tf u \\ \ttf{\beta} \end{bmatrix} = \begin{bmatrix} \tf R & \tf N & \tf R B_2 \\ \tf M & \tf L & I + \tf M B_2 \\ C_2 (zI-A)^{-1} & I & C_2 (zI-A)^{-1} B_2 \end{bmatrix} \begin{bmatrix} \ttf{\delta_x} \\ \ttf{\delta_y} \\ \ttf{\delta_u} \end{bmatrix}. \nonumber
\end{equation}
As can be seen, the controller implementation is once again internally stable only when the open loop plant is stable (if the realization $(A,C_2)$ is detectable). This structure thus shows the necessity of introducing and analyzing the effects of perturbations on the controller internal state $\ttf{\beta}$.


Of course, when the open loop system is stable, the controller structures illustrated below may be appealing as they are simpler and easier to implement. \rev{In fact, we can show that the controller structure in Figure \ref{fig:alt1} is an alternative realization of the internal model control principle (IMC) \cite{IMC1, IMC2} as applied to the Youla parameterization. Specifically, for open loop stable systems, the Youla parameter is given by $\tf Q = \tf K (I - \tf{P_{22}} \tf K)^{-1}$. As we show in Lemma \ref{lem:qkl} of Section \ref{sec:Youla_convex}, the Youla parameter $\tf Q$ is equal to the system response $\tf L$ for open loop stable systems. We then have}
\rev{
\begin{subequations} \label{eq:imc}
\begin{align}
\tf u &= \tf L \tf y - \tf M B_2 \tf u \label{eq:imc0}\\
&= \tf Q \tf y - \tf L  C_2 (zI-A)^{-1} B_2 \tf u \label{eq:imc1} \\
&= \tf Q \tf y - \tf Q \tf{P_{22}} \tf u \label{eq:imc2} \\
& = \tf Q (\tf y - \tf{P_{22}} \tf u), \label{eq:imc3}
\end{align}
\end{subequations} 
where \eqref{eq:imc1} is obtained by substituting $\tf M = \tf L  C_2 (zI-A)^{-1}$ from \eqref{eq:output_fb2} into \eqref{eq:imc0}. Equation \eqref{eq:imc3} is exactly IMC. Thus, we see that IMC is equivalent to our proposed parameterization (and the simplified representation shown in Figure \ref{fig:alt1}) for open loop stable systems. } 

\begin{figure}[ht!]
      \centering
      \subfigure[\rev{Internal Model Control}]{%
      \includegraphics[width=0.2\textwidth]{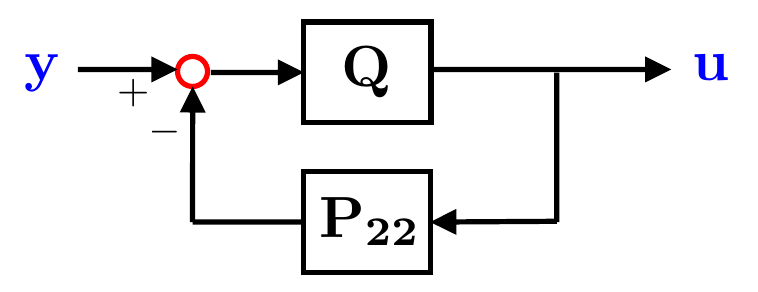}
      \label{fig:alt0}}
      
      \subfigure[Structure 1]{%
      \includegraphics[width=0.25\textwidth]{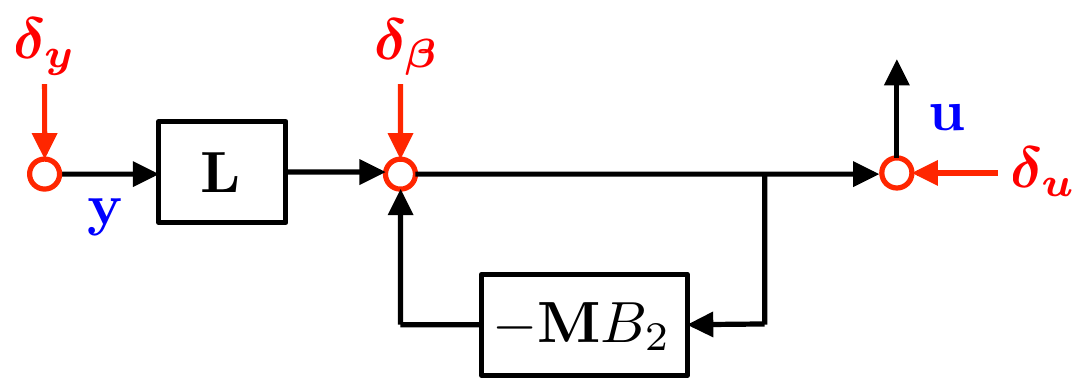}
      \label{fig:alt1}}
      
      \subfigure[Structure 2]{%
      \includegraphics[width=0.25\textwidth]{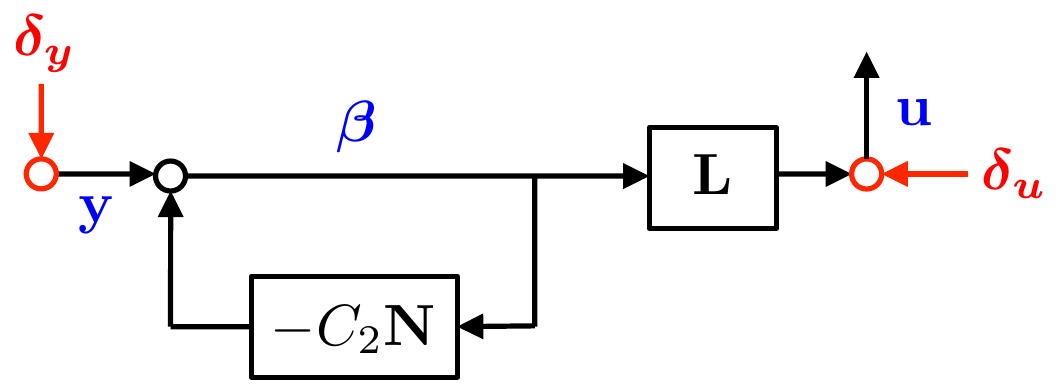}
      \label{fig:alt2}}
      
      \caption{Alternative controller structures for stable systems.}
\end{figure}

\subsection{Output Feedback with $D_{22} \not = 0$} \label{sec:non-sp}
Finally, for a general proper plant model \eqref{eq:dynamics} with $D_{22} \not = 0$, we define a new measurement $\bar{y}[t] = y[t] - D_{22} u[t]$. This leads to the controller structure shown in Figure \ref{fig:pof}. In this case, the closed loop transfer matrices from $\ttf{\delta_u}$ to the internal variables become
\begin{equation}
\begin{bmatrix} \tf x \\ \tf u \\ \tf y \\ \ttf{\beta} \end{bmatrix} = \begin{bmatrix} \tf R B_2 + \tf N D_{22} \\ I + \tf M B_2 + \tf L D_{22} \\ C_2 \tf R B_2 + D_{22} + C_2 \tf N D_{22} \\ -\frac{1}{z}B_2 (\tf M B_2 + \tf L D_{22}) \end{bmatrix} \ttf{\delta_u}. \nonumber
\end{equation}
The remaining entries of Table \ref{Table:1} remain the same. Therefore, the controller structure shown in Figure \ref{fig:pof} internally stabilizes the plant.
\begin{figure}[ht!]
      \centering
      \includegraphics[width=0.4\textwidth]{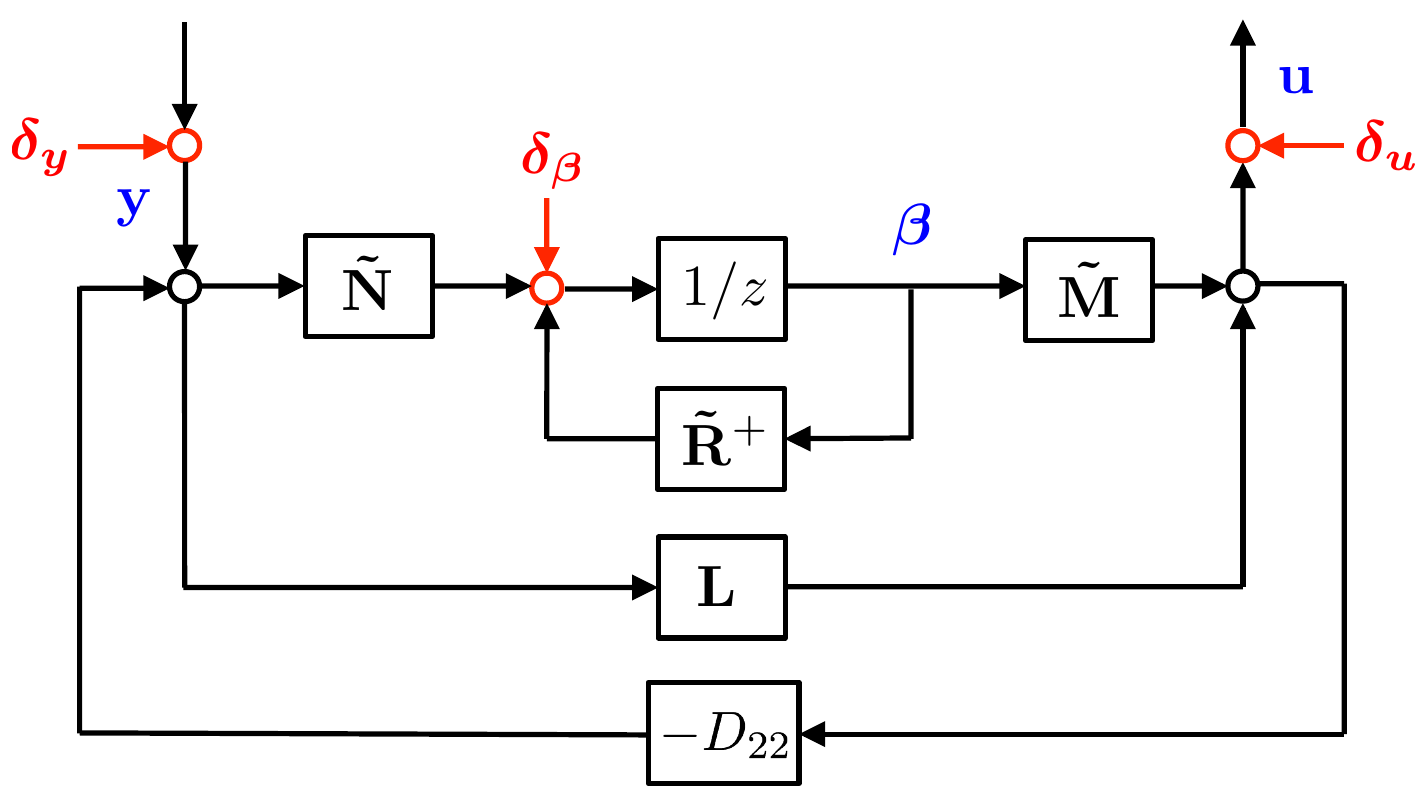}
      \caption{The proposed output feedback controller structure for $D_{22} \not = 0$.}
      \label{fig:pof}
\end{figure}

\subsection{\revsecond{System Level and Youla Parameterizations}}
\revsecond{A key difference between the SL and Youla parameterizations is the manner in which they characterize the achievable closed loop responses of a system.  The Youla parameterization provides an image space representation of the achievable system responses, parameterized explicitly by the free Youla parameter.  This parameterization lends itself naturally to efficient computation via the standard and theoretically supported approach \cite{Boyd_closed} of restricting the Youla parameter and objective function to be FIR.  However, despite this ease of computation, as alluded to earlier and discussed in detail in Section \ref{sec:sparsity}, imposing sparsity constraints on the controller via the Youla parameter is in general intractable.}

\revsecond{In contrast, the proposed SL parameterization specifies a kernel space representation of achievable system responses, parameterized implicitly by the affine space \eqref{eq:output_fb1} - \eqref{eq:output_fb2}.  While our discussion highlights the benefits and flexibility of the SL approach, there is the important caveat that the affine constraints \eqref{eq:output_fb1} - \eqref{eq:output_fb2} are in general infinite dimensional. Hence, although the parameterization is a convex one, it does not immediately lend itself to efficient computation.  In Section \ref{sec:FIRSLC}  we show that imposing FIR constraints on the system responses leads to a finite-dimensional optimization problem, and further show that such constraints are feasible if the system is controllable and observable.}

\section{System Level Constraints} \label{sec:class}
An advantage of the parameterizations described in the previous section is that they allow us to impose additional constraints on the system response \rev{and the corresponding internal structure of the controller}.  These constraints may be in the form of structural \rev{(subspace)} constraints on the response, or may capture a suitable measure of system performance: in this section, we provide a catalog of useful SLCs that can be naturally incorporated into the SLPs described in the previous section. In addition to all of the performance specifications described in \cite{Boyd_closed}, we also show that QI subspace constraints are a special case of SLCs.  We then provide an example as to why one may wish to go beyond QI subspace constraints to localized (sparse) subspace constraints on the system response, and show that such constraints can be trivially imposed in our framework.  \rev{As far as we are aware, no other parameterizations \cite{Polynomial,Factorization,Behavior_I,Behavior_II,Behavior_para,2006_Rotkowitz_QI_TAC} allow for such constraints to be tractably enforced for general (i.e., strongly connected) systems.  As such, we provide here a description of the largest known class of constrained stabilizing controllers that admit a convex parameterization. Further, as we show in our companion paper \cite{2015_PartII}, it is this ability to impose locality constraints on the controller structure via convex constraints that allows us to scale the methods proposed in \cite{Boyd_closed,2006_Rotkowitz_QI_TAC} to large-scale systems.}

\revsecond{Before proceeding, we emphasize that although the Youla parameterization and co-prime factors are needed to prove the results presented in Sections \ref{sec:Youla_convex} and \ref{sec:QI}, these are only used for the purposes of establishing connections between the Youla/QI parameterizations and the SLA.  The SLPs presented in the previous section require neither the Youla parameterization nor co-prime factors}.


\subsection{Constraints on the Youla Parameter} \label{sec:Youla_convex}
We show that any constraint imposed on the Youla parameter can be translated into a SLC, and vice versa.  In particular, if this constraint is convex, then so is the corresponding SLC.  Consider the following modification of the standard Youla parameterization, which characterizes a set of constrained internally stabilizing controllers $\tf K$ for a plant \eqref{eq:ofplant}:
\begin{equation}
\tf K =  (\tf Y_r - \tf U_r \tf Q)(\tf X_r - \tf V_r \tf Q)^{-1}, \ \tf Q \in \mathcal{Q} \cap \RHinf.
\label{eq:mod_youla}
\end{equation}
Here the expression for $\tf K$ is in terms of the co-prime factors defined in \rev{Section \ref{sec:move_Youla}}, 
and $\mathcal{Q}$ is an arbitrary set -- if we take $\mathcal{Q} = \RHinf$, we recover the standard Youla parameterization.  \rev{Similarly, if we take $\mathcal{Q}$ to be a QI subspace constraints, we recover a distributed optimal control problem that admits a convex parameterization: we discuss the connection between QI and SLCs in more detail in the next subsection.  Further, if the plant is open-loop stable or has special structure, it may be desirable to enforce non-QI constraints on the Youla parameter.  In general, one can use this expression to characterize all possible constrained internally stabilizing controllers by suitably varying the set $\mathcal{Q}$,\footnote{In particular, to ensure that $\tf K \in \mathcal{C}$, it suffices to enforce that $ (\tf Y_r - \tf U_r \tf Q)(\tf X_r - \tf V_r \tf Q)^{-1} \in \C$.} and hence this formulation is as general as possible.  We now show that an equivalent parameterization can be given in terms of a SLC.}

\begin{theorem} \label{thm:Youla2}
The set of constrained internally stabilizing controllers described by \eqref{eq:mod_youla} can be equivalently expressed as $\tf K = \tf L - \tf M \tf R^{-1}\tf N$, where the system response $\{\tf R, \tf M, \tf N, \tf L\}$ lies in the set
\begin{equation}
\{\tf R, \tf M, \tf N, \tf L \, \big{|} \, \text{ \eqref{eq:output_fb1} - \eqref{eq:output_fb3} hold, }  \tf L \in \mathfrak{M}(\mathcal Q)\},
\label{eq:SLC_Youla}
\end{equation} for $\mathfrak{M}(\tf Q) := \tf K(I- \tf P_{22} \tf K)^{-1} = (\tf Y_r - \tf U_r \tf Q) \tf U_l$ the invertible affine map as defined in \rev{Section \ref{sec:structured}}. Further, this parameterization is convex if and only if $\mathcal{Q}$ is convex.
\end{theorem}

In order to prove this result, we first need to understand the relationship between the controller $\tf K$, the Youla parameter $\tf Q$, and the system response $\{\tf R, \tf M, \tf N, \tf L\}$.  

\begin{lemma} \label{lem:qkl}
Let $\tf L$ be defined as in \eqref{eq:K_relation}, and the invertible affine map $\mathfrak{M}$ be defined as in \rev{Section \ref{sec:structured}}.  We then have that
\begin{equation}
\tf L = \tf K (I - \tf P_{22} \tf K)^{-1} = \mathfrak{M}(\tf Q).
\label{eq:qkl}
\end{equation}
\end{lemma}
\begin{proof}
From the equations $\tf u = \tf K \tf y$ and $\tf y = \tf P_{21} \tf w + \tf P_{22} \tf u$, we can eliminate $\tf u$ and express $\tf y$ as $\tf y = (I - \tf P_{22} \tf K)^{-1} \tf P_{21} \tf w$. We then have that
\begin{equation}
\tf u = \tf K \tf y = \tf K (I - \tf P_{22} \tf K)^{-1} \tf P_{21} \tf w. \label{eq:qithm}
\end{equation}
Recall that we define $\ttf{\delta_x} = B_1 \tf w$ and $\ttf{\delta_y} = D_{21} \tf w$. As a result, we have $\tf P_{21} \tf w = C_2(zI-A)^{-1} \ttf{\delta_x} + \ttf{\delta_y}$. Substituting this identity into \eqref{eq:qithm} yields
\begin{equation}
\tf u = \tf K (I - \tf P_{22} \tf K)^{-1} [ C_2(zI-A)^{-1} \ttf{\delta_x} + \ttf{\delta_y}]. \label{eq:qithm2}
\end{equation}
By definition, $\tf L$ is the closed loop mapping from $\ttf{\delta_y}$ to $\tf u$. Equation \eqref{eq:qithm2} then implies that $\tf L = \tf K (I - \tf P_{22} \tf K)^{-1}$. From\cite{2014_Sabau_QI, 2014_Lamperski_H2_journal} (c.f. \rev{Section \ref{sec:structured}}), we have $\tf K (I - \tf P_{22} \tf K)^{-1} = \mathfrak{M}(\tf Q)$, which completes the proof.
\end{proof}

\begin{proof}[Proof of Theorem \ref{thm:Youla2}]
The equivalence between the parameterizations \eqref{eq:mod_youla} and \eqref{eq:SLC_Youla} is readily obtained from Lemma \ref{lem:qkl}. As $\mathfrak{M}$ is an invertible affine mapping between $\tf L$ and $\tf Q$, any convex constraint imposed on the Youla parameter $\tf Q$ can be equivalently translated into a convex SLC imposed on $\tf L$, and vice versa. 
\end{proof}


\subsection{Quadratically Invariant Subspace Constraints} \label{sec:QI}
Recall that for a subspace $\mathcal{C}$ that is quadratically invariant with respect to a plant $\tf P_{22}$, the set of internally stabilizing controllers $\tf K$ that lie within the subspace $\mathcal{C}$ can be expressed as the set of stable transfer matrices $\tf Q \in \RHinf$ satisfying $\mathfrak{M}(\tf Q) \in \mathcal{C}$, for $\mathfrak{M}$ \rev{the invertible affine map defined in Section \ref{sec:structured}}. 
We therefore have the following corollary to Theorem \ref{thm:Youla2}.

\begin{corollary}
Let $\C$ be a subspace constraint that is quadratically invariant with respect to $\tf P_{22}$. Then the set of internally stabilizing controllers satisfying $\tf K \in \mathcal{C}$ can be parameterized as in Theorem \ref{thm:Youla2} with \rev{$\tf L = \mathfrak{M}(\tf{Q}) \in \mathcal{C}$}.
 \label{thm:qi}
\end{corollary}
\begin{proof}
From Lemma \ref{lem:qkl}, we have $\tf L = \tf K (I - \tf P_{22} \tf K)^{-1}$. Invoking Theorem $14$ of \cite{2006_Rotkowitz_QI_TAC}, we have that $\tf K \in \C$ if and only if $\tf L = \tf K (I - \tf P_{22} \tf K)^{-1} \in \C$. The claim then follows immediately from Theorem \ref{thm:Youla2}.
\end{proof}

Note that Corollary \ref{thm:qi} holds true for stable and unstable plants $\tf P$. Therefore, in order to parameterize the set of internally stabilizing controllers lying in $\mathcal{C}$, we do not need to assume the existence of an initial strongly stabilizing controller as in \cite{2006_Rotkowitz_QI_TAC} nor do we need to perform a doubly co-prime factorization as in \cite{2014_Sabau_QI}.  
Thus we see that QI subspace constraints are a special case of SLCs.  

Finally, we note that in \cite{2014_Lessard_convexity} and \cite{2011_QIN}, the authors show that QI is necessary for a subspace constraint $\mathcal{C}$ on the controller $\tf K$ to be enforceable via a convex constraint on the Youla parameter $\tf Q$. However, when $\C$ is not a subspace constraint, no general methods exist to determine whether the set of internally stabilizing controllers lying in $\C$ admits a convex representation.  In contrast, determining the convexity of a SLC is straightforward.  

\subsection{Beyond QI}\label{sec:beyond}
Before introducing the class of localized SLCs, we present a simple example for which the QI framework fails to capture an ``obvious'' controller with localized structure, but for which the SLA can.  This example also serves to illustrate the importance of locality in achieving scalability of controller implementation.  Our companion paper \cite{2015_PartII} shows how locality further leads to scalability of controller synthesis.

\begin{example}
\label{ex:motivating}
\rev{
Consider the optimal control problem:
\begin{equation}
\begin{array}{rl}
\minimize{u} & \lim_{T\to \infty}\frac{1}{T}\sum_{t=0}^T \mathbb{E}\|x[t]\|_2^2 \\
\text{subject to} & x[t+1]=Ax[t]+u[t]+w[t],
\end{array}
\label{eq:opt_ctrl_example}
\end{equation}
with disturbance $w[t]\overset{\mathrm{i.i.d}}{\sim{}}\mathcal{N}(0,I)$.  We assume full state-feedback, i.e., the control action at time $t$ can be expressed as $u[t]=f(x[0:t])$ for some function $f$.  An optimal control policy $u^\star$ for this LQR problem is easily seen to be given by $u^\star[t]=-Ax[t]$. } 

\rev{Further suppose that the state matrix $A$ is sparse and let its support define the adjacency matrix of a graph $\mathcal{G}$ for which we identify the $i$th node with the corresponding state/control pair $(x_i,u_i)$.  
In this case, we have that the optimal control policy $u^\star$ can be implemented in a \emph{localized} manner. In particular, in order to implement the state feedback policy for the $i$th actuator $u_i$, only those states $x_j$ for which $A_{ij}\neq 0$ need to be collected -- thus only those states corresponding to immediate neighbors of node $i$ in the graph $\mathcal{G}$, i.e., only \emph{local} states, need to be collected to compute the corresponding control action, leading to a localized implementation.  As we discuss in our companion paper \cite{2015_PartII}, the idea of locality is essential to allowing controller synthesis and implementation to scale to arbitrarily large systems, and hence such a structured controller is desirable.  }

\rev{Now suppose that we naively attempt to solve optimal control problem \eqref{eq:opt_ctrl_example} by converting it to its equivalent $\mathcal{H}_2$ model matching problem \rev{\eqref{eq:decentralized}} and constraining the controller $\tf K$ to have the same support as $A$, i.e., $\tf K = \sum_{t=0}^\infty \frac{1}{z^t}K[t]$, $\supp{K[t]} \subset \supp {A}$.} 
\revsecond{If the graph $\mathcal{G}$ is strongly connected, then \emph{any} sparsity constraint in the form of $\tf{K}_{ij} = 0$ is not QI with respect to the plant $\tf P_{22} = (zI-A)^{-1}$.  To see this, note that if the graph $\mathcal{G}$ is strongly connected, then $\tf P_{22}$ is a dense transfer matrix: it then follows immediately that any subspace $\mathcal{C}$ enforcing sparsity constraints on $\tf K$ fails to satisfy $\tf K \tf P_{22} \tf K \in \mathcal{C}, \,\, \forall \tf K \in \mathcal{C}$, and hence is not QI with respect to $\tf P_{22}$.
The results of \cite{2014_Lessard_convexity} further allow us to conclude that computing such a structured controller can never be done using convex programming when using the Youla parameterization.}

\rev{In contrast, in the case of a full control ($B_2=I$) problem, the condition \eqref{eq:state_fb} simplifies to $(zI-A)\tf R - \tf M = I$, $\tf R, \tf M \in \frac{1}{z}\RHinf$.   Again, suppose that we wish to synthesize an optimal controller that has a communication topology given by the support of $A$ -- from the above implementation, it suffices to constrain the support of transfer matrices $\tf R$ and $\tf M$ to be a subset of that of $A$. It can be checked that $\tf R = \frac{1}{z}I$, and $\tf M = -\frac{1}{z}A$ satisfy the above constraints, and recover the globally optimal controller $\tf K = -A$. } 
\end{example}

\subsection{Subspace and Sparsity Constraints} \label{sec:sparsity}
\rev{Motivated by the previous example, we consider here subspace SLCs, with a particular emphasis on those that encode sparse structure in the system response and corresponding controller implementation.} Let $\mathcal{L}$ be a subspace of $\RHinf$.  We can parameterize all stable achievable system responses that lie in this subspace by adding the following SLC to the parameterization of Theorem \ref{thm:of}:
\begin{equation}
\begin{bmatrix}
\tf R & \tf N \\
\tf M & \tf L \end{bmatrix} \in \mathcal{L}.
\label{eq:SLC_localized}
\end{equation}

Of particular interest are subspaces $\mathcal{L}$ that define transfer matrices of sparse support.  
An immediate benefit of enforcing such sparsity constraints on the system response is that implementing the resulting controller \eqref{eq:ss_like} can be done in a \emph{localized way}, i.e., each controller state $\beta_i$ and control action $u_i$ can be computed using a local subset (as defined by the support of the system response) of the global controller state $\beta$ and sensor measurements $y$.  For this reason, we refer to the constraint \eqref{eq:SLC_localized} as a \emph{localized} SLC when it defines a subspace with sparse support.  As we show in our companion paper \cite{2015_PartII}, such localized constraints further allow for the resulting system response to be computed in a localized way, i.e., the global computation decomposes naturally into decoupled subproblems that depend only on local sub-matrices of the state-space representation \eqref{eq:dynamics}.  Clearly, both of these features are extremely desirable when computing controllers for large-scale systems.  \rev{To the best of our knowledge, such constraints cannot be enforced using convex constraints using existing controller parameterizations \cite{Polynomial,Factorization,Behavior_I,Behavior_II,Behavior_para,2006_Rotkowitz_QI_TAC} for general systems.}

\rev{A caveat of our approach is that although arbitrary subspace structure can be enforced on the system response, it is possible that the intersection of the affine space described in Theorem \ref{thm:of} with the specified subspace is empty.  Indeed, selecting an appropriate (feasible) localized SLC, as defined by the subspace $\mathcal{L}$,  is a subtle task: it depends on an interplay between actuator and sensor density, information exchange delay and disturbance propagation delay. Formally defining and analyzing a procedure for designing a localized SLC is beyond the scope of this paper: as such, we refer the reader to our recent paper \cite{2015_Wang_Reg}, in which we present a method that allows for the joint design of an actuator architecture and corresponding feasible localized SLC. } 

\subsection{FIR Constraints} \label{sec:FIRSLC}
Given the parameterization of stabilizing controllers of Theorem \ref{thm:of}, it is straightforward to enforce that a system response be FIR with horizon $T$ via the following SLC
\begin{equation}
\tf R, \tf M, \tf N, \tf L \in \FT.
\label{eq:SLC_FIR}
\end{equation}

\rev{Whereas the pros and cons of deadbeat control in the centralized setting are well studied \cite{deadbeat1,deadbeat2,deadbeat3}, we argue here that imposing an appropriately tuned FIR SLC has benefits that are specific to the distributed large-scale setting:}
\begin{enumerate}[(a)]
  \item The controller achieving the desired system response can be implemented using the FIR filter banks $\tf{\tilde{R}^+}, \tf{\tilde{M}}, \tf{\tilde{N}}, \tf L \in \FT$, as illustrated in Figure \ref{fig:of}.  This simplicity of implementation is extremely helpful when applying these methods in practice.
  \item When a FIR SLC is imposed, the resulting set of stable achievable system responses and corresponding controllers admit a finite dimensional representation -- specifically, the constraints specified in Theorem \ref{thm:of} only need to be applied to the impulse response elements $\{R[t], M[t], N[t], L[t]\}_{t=0}^T$.
\end{enumerate}

\begin{remark}
It should be noted that the computational benefits claimed above hold only for discrete time systems. For continuous time systems, a FIR transfer matrix is still an infinite dimensional object, and hence the resulting parameterizations and constraints are in general infinite dimensional as well.
\end{remark}
\rev{\begin{remark} The complexity of local implementations using FIR filter banks scales linearly with the horizon $T$ -- an interesting direction for future work is to determine if infinite impulse response (IIR) system responses lead to simpler controller implementations via state-space realizations.
\end{remark}}

\revsecond{We conclude this subsection by showing that such FIR constraints are always feasible, for suitably chosen horizons $T$, if the system is controllable and observable.}
\revsecond{\begin{theorem} \label{thm:fd}
The SLP \eqref{eq:output_fb} admits a FIR solution if the triple $(A, B_2, C_2)$ is controllable and observable. 
\end{theorem}
\begin{proof}
By definition, if $(A,B_2)$ is controllable, then there exists FIR transfer matrices $(\tf{R_1}, \tf{M_1}) \in \mathcal{F_T}_1$ satisfying \eqref{eq:state_fb} for some finite $T_1$. Similarly, if $(A,C_2)$ is observable, then there exists FIR transfer matrices $(\tf{R_2}, \tf{N_2}) \in \mathcal{F_T}_2$ satisfying \eqref{eq:state_est} for some finite $T_2$. When $(A, B_2, C_2)$ is controllable and observable, the following FIR transfer matrices can be verified to lie in the affine space \eqref{eq:output_fb} 
\begin{subequations}
\begin{align}
\tf R &= \tf R_1 + \tf R_2 - \tf R_1 (zI-A) \tf R_2 \label{eq:sfest1}\\
\tf M &= \tf M_1 - \tf M_1 (zI-A) \tf R_2 \label{eq:sfest2}\\
\tf N &= \tf N_2 - \tf R_1 (zI-A) \tf N_2 \label{eq:sfest3}\\
\tf L &= -\tf M_1 (zI-A) \tf N_2. \label{eq:sfest4}
\end{align}\label{eq:sfest}
\end{subequations}
\end{proof}}

\revsecond{Finally, we note that recently developed relaxations \cite{Virtual} of SLP can be used when such FIR constraints cannot be satisfied. This may occur, for instance, when the underlying system is only stabilizable and/or detectable.}

\subsection{Intersections of SLCs and Spatiotemporal Constraints} \label{eq:spatio}
Another major benefit of SLCs is that several such constraints can be imposed on the system response at once.  Further, as convex sets are closed under intersection, convex SLCs are also closed under intersection.  To illustrate the usefulness of this property, consider the intersection of a QI subspace SLC (enforcing information exchange constraints between sub-controllers), a FIR SLC and a localized SLC.  The resulting SLC can be interpreted as enforcing a spatiotemporal constraint on the system response and its corresponding controller, as we explain using the chain example \rev{shown below}. 

Figure \ref{fig:st1} shows a diagram of the system response to a particular disturbance $(\ttf{\delta_{x}})_i$.  In this figure, the vertical axis denotes the spatial coordinate of a state in the chain, and the horizontal axis denotes time: hence we refer to this figure as a space-time diagram. Depicted are the three components of the spatiotemporal constraint, namely the communication delay imposed on the controller via the QI subspace SLC, the deadbeat response of the system to the disturbance imposed by the FIR SLC, and the localized region affected by the disturbance $(\ttf{\delta_{x}})_i$ imposed by the localized SLC.

When the effect of each disturbance $(\ttf{\delta_{x}})_i$ can be localized within such a spatiotemporal SLC, the system is said to be \emph{localizable} (c.f., \cite{2014_Wang_CDC, 2015_Wang_H2}).
\revsecond{It follows that the feasibility of a spatiotemporal constraint implies a more general notion of controllability (observability), wherein the system impulse response is constrained to be finite in both space and time, and the controller is subject to communication delays.  Thus rather than the traditional computational test of verifying the rank of a suitable controllability (observability) matrix, localizability is verified by the feasibility of a set of affine constraints.}

\begin{figure}[h!]
      \centering
      \includegraphics[width=0.35\textwidth]{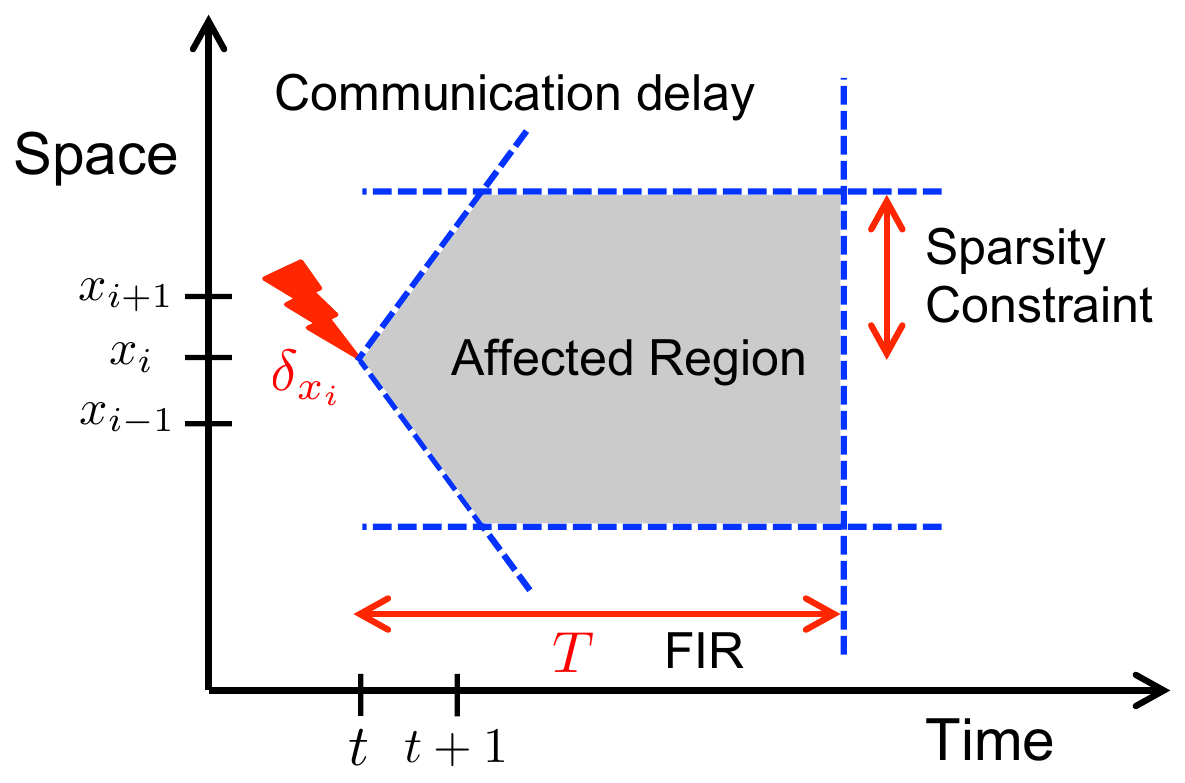}
      \caption{Space time diagram for a single disturbance striking the chain described in Example \ref{ex:motivating}.}
      \label{fig:st1}
\end{figure}

\subsection{Closed Loop Specifications} \label{sec:cl_spec}
\rev{As in \cite{Boyd_closed}, our parameterization allows for arbitrary performance constraints to be imposed on the closed loop response.  In contrast to the method proposed in \cite{Boyd_closed}, these performance constraints can be combined with structural (i.e., localized spatiotemporal) constraints on the controller realization, naturally extending their applicability to the large-scale distributed setting. In the interest of completeness, we highlight some particularly useful SLCs here.}

\subsubsection{System Performance Constraints} \label{sec:sys_perf}
Let $g(\cdot)$ be a functional of the system response --  it then follows that all internally stabilizing controllers satisfying a  performance level, as specified by a scalar $\gamma$, are given by transfer matrices $\{\tf R, \tf M, \tf N, \tf L\}$  satisfying the conditions of Theorem \ref{thm:of} and the SLC
\begin{equation}
g(\tf R, \tf M, \tf N, \tf L ) \leq \gamma.
\label{eq:perf_SLC}
\end{equation}

Further, recall that the sublevel set of a convex functional is a convex set, and hence if $g$ is convex, then so is the SLC \eqref{eq:perf_SLC}.  A particularly useful choice of convex functional is
\begin{equation}
g(\tf R, \tf M, \tf N, \tf L ) = \left\| \begin{bmatrix} C_1 & D_{12} \end{bmatrix} \begin{bmatrix} \tf R & \tf N \\ \tf M & \tf L \end{bmatrix} \begin{bmatrix} B_1 \\ D_{21} \end{bmatrix} + D_{11} \right\|,
\label{eq:io_SLC}
\end{equation}
for a system norm $\|\cdot\|$, which is equivalent to the objective function of the decentralized optimal control problem \eqref{eq:decentralized}.  Thus by imposing several performance SLCs \eqref{eq:io_SLC} with different choices of norm, one can naturally formulate multi-objective optimal control problems.

\subsubsection{Controller Robustness Constraints}
Suppose that the controller is to be implemented using limited hardware, thus introducing non-negligible quantization (or other errors) to the internally computed signals: this can be modeled via an internal additive noise $\ttf{\delta_\beta}$ in the controller structure (c.f., Figure \ref{fig:of}).  In this case, we may wish to design a controller that further limits the effects of these perturbations on the system: to do so, we can impose a performance SLC on the closed loop transfer matrices specified in the rightmost column of Table \ref{Table:1}.  

\subsubsection{Controller Architecture Constraints}\label{sec:rfd_slc}
The controller implementation \eqref{eq:ss_like} also allows us to naturally control the number of actuators and sensors used by a controller -- this can be useful when designing controllers for large-scale systems that use a limited number of hardware resources (c.f., Section \ref{sec:RFD}).  In particular, assume that implementation \eqref{eq:ss_like} parameterizing stabilizing controllers that use all possible actuators and sensors. It then suffices to constrain the number of non-zero rows of the transfer matrix $[\tf{\tilde{M}}, \tf L]$ to limit the number of actuators used by the controller, and similarly, the number of non-zero columns of the transfer matrix $[\tf{\tilde{N}}^\top, \tf{L}^\top]^\top$ to limit the number of sensors used by the controller.  As stated, these constraints are non-convex, but recently proposed convex relaxations \cite{Matni_RFD_TAC,Matni_RFD_CDC} can be used in their stead to impose convex SLCs on the controller architecture.

\subsubsection{Positivity Constraints}
It has recently been observed that (internally) positive systems are amenable to efficient analysis and synthesis techniques (c.f., \cite{2015_Positive} and the references therein).  Therefore it may be desirable to synthesize a controller that either preserves or enforces positivity of the resulting closed loop system.  We can enforce this condition via the SLC that the elements
\begin{equation}
\Big\{ \begin{bmatrix} C_1 & D_{12} \end{bmatrix} \begin{bmatrix} R[t] & N[t] \\ M[t] & L[t]\end{bmatrix} \begin{bmatrix} B_1 \\ D_{21} \end{bmatrix} \Big\}_{t=1}^\infty \nonumber
\end{equation}
and the matrix $(D_{12} L[0] D_{21} + D_{11})$ are all element-wise nonnegative matrices. This SLC is easily seen to be convex.




\section{System Level Synthesis} \label{sec:localizability}
We build on the results of the previous sections to formulate the SLS problem. We show that by combining appropriate SLPs and SLCs, the largest known class of convex structured optimal control problems can be formulated.  As a special case, we show that we recover all possible structured optimal control problems of the form \eqref{eq:decentralized} that admit a convex representation in the Youla domain.

\subsection{General Formulation}
Let $g(\cdot)$ be a functional capturing a desired measure of the performance of the system (as described in Section \ref{sec:sys_perf}), and let $\mathcal{S}$ be a SLC.  We then pose the SLS problem as
\begin{eqnarray}
\underset{\{\tf R,\tf M,\tf N,\tf L\}}{\text{minimize    }} && g(\tf R,\tf M,\tf N,\tf L) \nonumber\\
\text{subject to } && \eqref{eq:output_fb1} - \eqref{eq:output_fb3} \nonumber \\
&& \begin{bmatrix} \tf R & \tf N \\ \tf M & \tf L \end{bmatrix} \in \mathcal{S}. \label{eq:main_SS}
\end{eqnarray}

For $g$ a convex functional and $\mathcal{S}$ a convex set,\footnote{More generally, we only need the intersection of the set $\mathcal{S}$ and the restriction of the functional $g$ to the affine subspace described in \eqref{eq:output_fb} to be convex.} the resulting SLS problem is a convex optimization problem.


\begin{remark}
For a state feedback problem, the SLS problem can be simplified to
\begin{eqnarray}
\underset{\{\tf R,\tf M\}}{\text{minimize    }} && g(\tf R,\tf M) \nonumber\\
\text{subject to } && \eqref{eq:state_fb1} - \eqref{eq:state_fb2} \nonumber \\
&& \begin{bmatrix} \tf R \\ \tf M \end{bmatrix} \in \mathcal{S}. \label{eq:main_SSF}
\end{eqnarray}
\end{remark}

\subsection{\revsecond{Examples of Convex SLS}}
\revsecond{Here we highlight some convex SLS problems. A more extensive list can be found in \cite{Thesis,SLStutorial}.}
\subsubsection{\revsecond{Distributed Optimal Control}} 
\revsecond{The distributed optimal control problem \eqref{eq:decentralized} with a QI subspace constraint $\mathcal{C}$ can be formulated as a SLS problem as}
\revsecond{
\begin{eqnarray}
\text{minimize    } && \eqref{eq:io_SLC} \nonumber\\
\text{subject to } && \eqref{eq:output_fb1} - \eqref{eq:output_fb3}, \,\,  \tf L \in \mathcal{C}.
\end{eqnarray}
}
\revsecond{Thus all distributed optimal control problems that can be formulated as convex optimization problems in the Youla domain are special cases of convex SLS problem \eqref{eq:main_SS}.}

\subsubsection{Localized LQG Control}\label{sec:local}
In \cite{2014_Wang_CDC,2015_Wang_H2} we posed and solved a localized LQG optimal control problem.  
In the case of a state-feedback problem \cite{2014_Wang_CDC}, the resulting SLS problem is of the form
\begin{eqnarray}
\underset{\{\tf R,\tf M\}}{\text{minimize    }} && \|C_1 \tf R + D_{12} \tf M \|_{\mathcal{H}_2}^2 \nonumber\\
\text{subject to } && \eqref{eq:state_fb1} - \eqref{eq:state_fb2} \nonumber \\
&& \begin{bmatrix} \tf R \\ \tf M \end{bmatrix} \in \mathcal{C}\cap\mathcal{L}\cap\mathcal{F}_T \label{eq:LLQR},
\end{eqnarray}
for $\mathcal{C}$ a QI subspace SLC, $\mathcal{L}$ a sparsity SLC, and $\mathcal{F}_T$ a FIR SLC.  

The observation that we make in \cite{2014_Wang_CDC} (and extend to the output feedback setting in \cite{2015_Wang_H2}), is that the localized SLS problem \eqref{eq:LLQR} can be decomposed into a set of independent sub-problems solving for the columns $\tf R_i$ and $\tf M_i$ of the transfer matrices $\tf R$ and $\tf M$ -- as these problems are independent, they can be solved in parallel. Further, the sparsity constraint $\mathcal{L}$ restricts each sub-problem to a local subset of the system model and states, as specified by the nonzero components of the corresponding column of the transfer matrices $\tf R$ and $\tf M$ (e.g., as was described in Example \ref{ex:motivating}), allowing each of these sub-problems to be expressed in terms of optimization variables (and corresponding sub-matrices of the state-space realization \eqref{eq:output_fb}) that are of significantly smaller dimension than the global system response $\{ \tf R, \tf M\}$.  Thus for a given feasible spatiotemporal SLC, the localized SLS problem \eqref{eq:LLQR} can be solved for arbitrarily large-scale systems, assuming that each sub-controller can solve its corresponding sub-problem in parallel.\footnote{We also show how to co-design an actuation architecture and feasible corresponding spatiotemporal constraint in \cite{2015_Wang_Reg}, and so the assumption of a feasible spatiotemporal constraint is a reasonable one.} \rev{As far as we are aware, such constrained optimal control problems cannot be solved via convex programming using existing controller parameterizations in the literature.}

In our companion paper \cite{2015_PartII}, we generalize all of these concepts to the system level approach to controller synthesis, and show that appropriate notions of separability for SLCs can be defined which allow for optimal controllers to be synthesized and implemented with order constant complexity (assuming parallel computation is available for each subproblem) relative to the global system size.
 
\subsubsection{\revsecond{Regularization for Design}}\label{sec:RFD}
\revsecond{The regularization for design framework (RFD) \cite{Matni_RFD_TAC,Matni_RFD_CDC,Matni_Comms_TCNS,Matni_Comms_CDC} explores tradeoffs between closed loop performance and architectural cost using convex programming by augmenting the objective function with a suitable convex regularizer that penalizes the use of actuators, sensors and communication links. To integrate RFD into the SLA, it suffices to add a suitable convex regularizer, as mentioned in Section \ref{sec:rfd_slc} and described in \cite{Matni_RFD_TAC, 2015_Wang_Reg}, to the objective function of the SLS problem \eqref{eq:main_SS}.  We demonstrate the usefulness of combining RFD, locality and SLS in our companion paper \cite{2015_PartII}.}


\subsection{Computational Complexity and Non-convex Optimization}
A final advantage of the SLS problem \eqref{eq:main_SS} is that it is transparent to determine the computational complexity of the optimization problem. Specifically, the complexity of solving \eqref{eq:main_SS} is determined by the type of the objective function $g(\cdot)$ and the characterization of the intersection of the set $\mathcal{S}$ and the affine space \eqref{eq:output_fb1} - \eqref{eq:output_fb3}.  Further, when the SLS problem is non-convex, the direct nature of the formulation makes it straightforward to determine suitable convex relaxations or non-convex optimization techniques for the problem. In contrast, as discussed in \cite{2014_Lessard_convexity}, no general method exists to determine the computational complexity of the decentralized optimal control problem \eqref{eq:decentralized} for a general constraint set $\C$.

\section{Conclusion} \label{sec:conclusion}
In this paper, we defined and analyzed the system level approach to controller synthesis, which consists of three elements: System Level Parameterizations (SLPs), System Level Constraints (SLCs), and System Level Synthesis (SLS) problems.  We showed that all achievable and stable system responses can be characterized via the SLPs given in Theorems \ref{thm:sf} and \ref{thm:of}.  We further showed that these system responses could be used to parameterize internally stabilizing controllers that achieved them, and proposed a novel controller implementation \eqref{eq:ss_like}.  We then argued that this novel controller implementation had the important benefit of allowing for SLCs to be naturally imposed on it, and showed in Section \ref{sec:class} that using this controller structure and SLCs, we can characterize the broadest known class of constrained internally stabilizing controllers that admit a convex representation.  Finally, we combined SLPs and SLCs to formulate the SLS problem, and showed that it recovered as a special case many well studied constrained optimal controller synthesis problems from the literature.  In our companion paper \cite{2015_PartII}, we show how to use the system level approach to controller synthesis to co-design controllers, system responses and actuation, sensing and communication architectures for large-scale networked systems.


%

\appendices
\section{Stabilizability and Detectability} \label{sec:proof}

\begin{lemma}
The pair $(A, B_2)$ is stabilizable if and only if the affine subspace defined by \eqref{eq:state_fb} is non-empty.
 \label{lemma:1}
\end{lemma}
\begin{proof}
We first show that the stabilizability of $(A,B_2)$ implies that there exist transfer matrices $\tf R, \tf M\in\frac{1}{z}\RHinf$ satisfying equation \eqref{eq:state_fb1}. From the definition of stabilizability, there exists a matrix $F$ such that $A+B_2 F$ is a stable matrix. Substituting the state feedback control law $u = F x$ into \eqref{eq:zsfb}, we have $\tf x = (zI-A-B_2 F)^{-1}  \ttf{\delta_x}$ and $\tf u = F(zI-A-B_2 F)^{-1}  \ttf{\delta_x}$. 
The system response is given by $\tf R = (zI-A-B_2 F)^{-1}$ and $\tf M = F(zI-A-B_2 F)^{-1}$, which lie in $\frac{1}{z}\RHinf$ and are a solution to \eqref{eq:state_fb1}.

For the opposite direction, we note that $\tf R, \tf M \in \RHinf$ implies that these transfer matrices do not have poles outside the unit circle $| z | \geq 1$. From \eqref{eq:state_fb1}, we further observe that $\begin{bmatrix} zI - A & -B_2 \end{bmatrix}$ is right invertible in the region where $\tf R$ and $\tf M$ do not have poles, with $\begin{bmatrix} \tf R^\top & \tf M^\top \end{bmatrix}^\top$ being its right inverse. This then implies that $\begin{bmatrix} zI - A & -B_2 \end{bmatrix}$ has full row rank for all $| z | \geq 1$. This is equivalent to the PBH test \cite{Dullerud_Paganini} for stabilizability, proving the claim.
\end{proof}

We note that the analysis for the state feedback problem in Section \ref{sec:sf} can be applied to the state estimation problem by considering the dual to a full control system (c.f.,  \S 16.5 in \cite{Zhou1996robust}). For instance, the following corollary to Lemma \ref{lemma:1} gives an alternative definition of the detectability of pair $(A,C_2)$ \cite{2015_Wang_LDKF}.
\begin{corollary}
The pair $(A, C_2)$ is detectable if and only if the following conditions are feasible:
\begin{subequations} \label{eq:state_est}
\begin{align}
& \begin{bmatrix} \tf R & \tf N \end{bmatrix} \begin{bmatrix} zI - A \\ -C_2 \end{bmatrix} = I \label{eq:state_est1}\\
& \tf R, \tf N \in \frac{1}{z} \mathcal{RH}_\infty. \label{eq:state_est2}
\end{align}
\end{subequations} \label{cor:1}
\end{corollary}

A parameterization of all detectable observers can be constructed using the affine subspace \eqref{eq:state_est} in a manner analogous to that described above.

\begin{lemma}\label{lem:stab_det}
The triple $(A, B_2, C_2)$ is stabilizable and detectable if and only if the affine subspace described by \eqref{eq:output_fb} is non-empty.
\end{lemma}
\begin{proof}
\revsecond{This follows from an identical construction as that presented in the proof Theorem \ref{thm:fd}, but now using stable transfer matrices with possibly infinite impulse responses.}
%
\end{proof}

\ifCLASSOPTIONcaptionsoff
  \newpage
\fi



\bibliographystyle{IEEEtran}
\bibliography{Distributed}
%
%
%

%

\vspace{-10mm}
\begin{IEEEbiography}[{\includegraphics[width=1in,height=1.25in,clip,keepaspectratio]{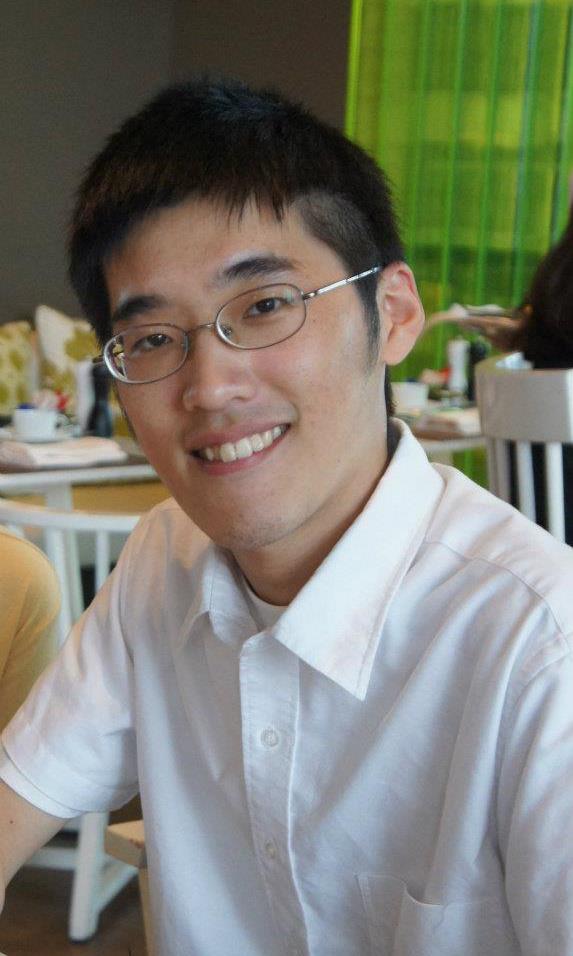}}]
{Yuh-Shyang Wang} (M'10) received the B.S. degree in electrical engineering from National Taiwan University, Taipei, Taiwan, in 2011, and the Ph.D. degree in control and dynamical systems from Caltech, Pasadena, CA, USA, in 2016 under the advisement of John C. Doyle.

He is currently a Research Engineer at GE Global Research Center, Niskayuna, NY, USA. His research interests include optimization, control, and machine learning for industrial cyber-physical systems and renewable energy systems.

Dr. Wang was the recipient of the 2017 ACC Best Student Paper Award.
\end{IEEEbiography}
\vspace{-10mm}
\begin{IEEEbiography}[{\includegraphics[width=1in,height=1.25in,clip,keepaspectratio]{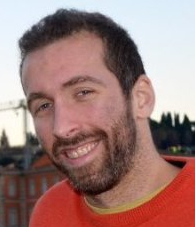}}]
{Nikolai Matni} (M'08) received the B.A.Sc. and M.A.Sc. degrees in electrical engineering from the University of British Columbia, Vancouver, BC, Canada, in 2008 and 2010, respectively, and the Ph.D. degree in control and dynamical systems from the California Institute of Technology, Pasadena, CA, USA, in June 2016 under the advisement of John C. Doyle.

He is currently a Postdoctoral Scholar at Electrical Engineering \& Computer Sciences, UC Berkeley, Berkeley, CA, USA. His research interests include the use of learning, layering, dynamics, control and optimization in the design and analysis of complex cyber-physical systems.

Dr. Matni received the IEEE CDC 2013 Best Student Paper Award, the IEEE ACC 2017 Best Student Paper Award (as co-advisor), and was an Everhart Lecture Series speaker at Caltech.
\end{IEEEbiography}
%
\vspace{-10mm}
\begin{IEEEbiography}[{\includegraphics[width=1in,height=1.25in,clip,keepaspectratio]{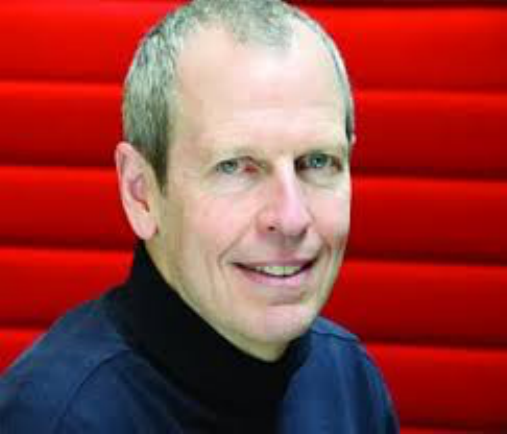}}]{John C. Doyle} received the B.S. and M.S. degrees in electrical engineering from Massachusetts Institute of Technology, Cambridge, MA, USA, in 1977, and the Ph.D. degree in mathematics from UC Berkeley, Berkeley, CA, USA, in 1984.

He is currently the Jean-Lou Chameau Professor of Control and Dynamical Systems, Electrical Engineer, and Bio-Engineering, Caltech, Pasadena, CA, USA. His research interests include mathematical foundations for complex networks with applications inbiology, technology, medicine, ecology, neuroscience, and multiscale physics that integrates theory from control, computation, communication, optimization, statistics (e.g., machine learning).

Dr. Doyle received the 1990 IEEE Baker Prize (for all IEEE publications), also listed in the world top 10 “most important” papers in mathematics 1981–1993, the IEEE Automatic Control Transactions Award (twice 1998, 1999), the 1994 AACC American Control Conference Schuck Award, the 2004 ACM Sigcomm Paper Prize and 2016 test of Time Award, and inclusion in Best Writing on Mathematics 2010. His individual awards include 1977 IEEE Power Hickernell, 1983 AACC Eckman, 1984 UC Berkeley Friedman, 1984 IEEE Centennial Outstanding Young Engineer (a one-time award for IEEE 100th anniversary), and 2004 IEEE Control Systems Field Award.
\end{IEEEbiography}




\end{document}